%% file: techreport.tex
\documentclass{sig-alternate}
\pdfoutput=1

\usepackage{graphicx}
\usepackage{balance}  
\usepackage{times,amsmath,epsfig,epstopdf,algorithm}
\usepackage{algorithmic}
\usepackage{amssymb} 
\usepackage{listings} 
\usepackage{subcaption}
\usepackage{caption}
\usepackage{graphics} 
\usepackage{color}

\newtheorem{thm}{Theorem}[section] 
\newtheorem{defn}{Definition}

\newtheorem{exmp}{Example}
\graphicspath{ {./Charts/} } 
\newcommand{\eat}[1]{}

\makeatletter
\renewcommand{\fnum@figure}{Fig. \thefigure}
\makeatother

\begin{document}
\title{Supporting Window Analytics over Large-scale Dynamic Graphs}

\numberofauthors{4}
\author{%
{Qi Fan{\small $~^{\natural}$}, Zhengkui Wang{\small $~^{\S}$}, Chee-Yong Chan{\small $~^{\sharp}$} }, Kian-Lee Tan{\small $~^{\sharp\natural}$}
\vspace{1.6mm}\\
\fontsize{12}{12}\selectfont\itshape
$^{\natural}$\, NUS Graduate School of Integrative Science and Engineering, National University of Singapore\\
\fontsize{12}{12}\selectfont\itshape
$^{\sharp}$\, School of Computing, National University of Singapore\\
\fontsize{12}{12}\selectfont\itshape
$^{\S}$\, Singapore Institute of Technology\\
}

\maketitle

\begin{abstract}
In relational DBMS, window functions have been widely used to facilitate data analytics. Surprisingly, while similar concepts have been employed for graph analytics, there has been no explicit notions of graph window analytic functions. In this paper, we formally introduce window queries for graph analytics. In such queries, for each vertex, the analysis is performed on a window of vertices defined based on the graph structure.
In particular, we identify two instantiations, namely the {\em k-hop window} and the {\em topological window}. 
We develop two novel indices, {\em Dense Block} index ({\em DBIndex}) and {\em Inheritance} index ({\em I-Index}), to facilitate efficient processing of these two types of windows respectively. Extensive experiments are conducted over both real and synthetic datasets with hundreds of  millions of vertices and edges. Experimental results indicate that our proposed index-based query processing solutions achieve four orders of magnitude of query performance gain than the non-index algorithm and are superior over EAGR\cite{mondal2014eagr} wrt scalability and efficiency.    
\end{abstract}

\category{H.2, E.5}{Database}{Optimization}


\input{sec1_introduction}
\input{sec2_related_work}

\input{sec3_definition}

\input{sec4_DBIndex}

\input{sec5_parent_index}
\input{sec6_experiment}

\input{sec7_conclusion}
\section{Acknowledgment}
Qi Fan is supported by NGS Scholarship. This work is supported by the MOE/NUS grant R-252-000-500-112 and AWS in Education Grant award. 

\bibliographystyle{abbrv}
\bibliography{citations} 
\end{document}

%% file: sec1_introduction.tex
\section{Introduction}
Information networks such as social networks, biological networks and
phone-call networks are typically modeled as graphs \cite{chen2008graph}
where the vertices correspond to objects and the edges
capture the relationships between these objects.
For instance, in social networks, every user is represented by
a vertex and the friendship between two users is reflected by an edge between
the vertices. In addition, a user's profile can be maintained as
the vertex's attributes. Such graphs contain a wealth of valuable 
information which can be analyzed to discover interesting patterns. 
For example, we can find the top-k influential users who can 
reach the most number of friends within 2 hops. With increasingly
larger network sizes, it is becoming significantly challenging to 
query, analyze and process these graph data. Therefore, there is an urgent need 
to develop effective and efficient mechanisms over graph data to draw out
information from such data resources.
 
Traditionally, in relational DBMS, window functions have been commonly
used for data analytics \cite{cao2012optimization, bellamkonda2013adaptive}. Instead of performing analysis (e.g. ranking, 
aggregate) over the entire data set, a window function returns for each 
input tuple a value derived from applying the function over a window of 
neighboring tuples. For instance, users may be interested in finding 
each employee's salary ranking within the department. Here,
each tuple's neighbors are essentially records from the same department.


Interestingly, the notion of window functions turns out to be not uncommon
in graph data. For instance, in a social network, it is important to detect 
a person's social position and influence among his/her social community. 
The ``social community'' of the person is essentially his/her ``window'' 
comprising neighbors derived from his/her k-hop friends.
However, as illustrated in this example, the structure of a graph
plays a critical role in determining the neighboring data of a vertex.
In fact, it is often useful to quantify a structural range to each vertex 
and then perform analytics over the range. 
Surprisingly, though such a concept of window functions has been widely
used, the notion has not been explicitly formulated. 
In this paper, we are motivated to extend the window queries in traditional 
SQL for supporting graph analysis. However, the window definition in 
SQL is no longer applicable in a graph context, as it does not capture 
the graph structure information.
Thus, we seek to formulate the notion of graph windows and to develop
efficient algorithms to process them over large scaled graph structures. 

We have identifed two instantiations of graph windows, namely 
$k${\em -hop} and {\em topological} windows. 
We first demonstrate these window semantics with the following examples. 
\begin{exmp}
\label{query:linkedin-2-hop-window}
\emph{(k-hop window)} In a social network ( such as Linked-In and Facebook etc.), users are normally modeled as vertices and connectivity relationships are modeled as edges. In social network scenario, it is of great interest to summarize the most relevant connections to each user such as the neighbors within 2-hops. Some analytic queries such as summarizing the related connections' distribution among different companies, and computing age distribution of the related friends can be useful. In order to answer these queries, collecting data from every user's neighborhoods within 2-hop is necessary.
\end{exmp}

\begin{exmp}
\label{query:bio-dag-window}
\emph{(Topological window)} In biological networks ( such as Argocyc, Ecocyc etc.\cite{keseler2005ecocyc}), genes, enzymes and proteins are vertices and their dependency in a pathway are edges. Because these networks are directed 
and acyclic, in order to study the protein regulating process, one may be 
interested to find out the statistics of molecules in each protein 
production pathway. For each protein, we can traverse the graph to find 
every other molecule that is in the upstream of its pathway. Then we can group and count the number of genes and enzymes among those molecules. 
\end{exmp}

A common feature among these examples is that data aggregation is needed based on a set of vertices (which is the {\em graph window}) 
defined according to each vertex.  To illustrate, in example \ref{query:linkedin-2-hop-window}, every user needs to gather data from its friends and friends-of-friends. 
The \emph{2-hop neighbors} form its window. Likewise, in example \ref{query:bio-dag-window}, every protein needs to count the number of particular type of genes preceding it in the regulating pathway. For every protein, the set of
\emph{preceding molecules} forms its window. 

To support the analyses in the above-mentioned examples, we propose a new 
type of query, \emph{Graph Window Query} (GWQ in short),
over a data graph. \emph{GWQ} is defined with respect to a graph structure 
and is important in a graph context. Unlike the traditional window in SQL, 
we identify two types of useful graph windows according to the 
graph structures, namely k-hop Window $W_{kh}$ and Topological Window $W_t$. 
A k-hop window forms a window for one vertex by using its k-hop neighbors. 
k-hop neighbors are important to one vertex, as these are the vertices 
showing structural closeness as in Example 1. The k-hop neighbors window 
we define here is similar to the egocentric-network in network analysis \cite{burt2009structural} \cite{mondal2014eagr}. A topological window, on the
other hand, forms a window for one vertex by using all its preceding 
vertices in a directed acyclic graph. The preceding vertices of one vertex are normally those which influence the vertex in a network as illustrated in Example 2. 

To the best of our knowledge, existing graph databases or graph query languages do not directly support our proposed GWQ. There are two major challenges in processing GWQ. First, we need an efficient scheme to  calculate the window of each vertex. Second, we need
efficient solutions to process the aggregation over a large number 
of windows that may overlap. This offers opportunities to share the 
computation. However, it is non-trivial to address these two challenges.  

For $k$-hop window like query, the state-of-the-art processing algorithm
is EAGR \cite{mondal2014eagr}. EAGR builds an overlay graph including 
the shared components of different windows. This is done 
in multiple iterations, each of which performs the following.
First, EAGR sorts all the vertices according to their $k$-hop 
neighbors based on their lexicographic order. 
Second, the sorted vertices are split into equal sized chunks each of which is further built as one frequent-pattern tree to mine the shared components. 
However, EAGR requires all the vertices' $k$-hop 
neighbors to be pre-computed and resided in memory during 
each sorting and mining operation;
otherwise, EAGR incurs high computation overhead if the pre-computed structure needs to be shuffled to/from disk.
This limits the efficiency and scalability of EAGR.
 For instance, 
a LiveJournal social network graph\footnote{Available at http://snap.stanford.edu/data/index.html, which is used in \cite{mondal2014eagr}} 
(4.8M vertices, 69M edges) generates over 100GB neighborhood information 
for k=2 in adjacency list representation. In addition, the overlay graph construction is not a one-time task,
but periodically performed after a certain number of structural updates in order to maintain the overlay quality. The high memory consumption renders the scheme impractical 
when $k$ and the graph size increases.

In this paper,
we propose \textit{Dense Block Index (DBIndex)} 
to process queries efficiently.
Like EAGR, DBIndex seeks to exploit common
components among different windows to salvage 
partial work done. However, unlike EAGR,
we identify the window similarity utilizing a hash-based 
clustering technique instead of sorting. This ensures 
efficient memory usage, as the window information of each vertex can 
be computed on-the-fly. On the basis of the clusters, we develop different optimizations 
to extract the shared components which result in an efficient index construction. 

Moreover, we provide another \emph{Inheritance Index (I-Index)} tailored 
to topological window query. I-Index differentiates itself from
DBIndex by integrating more descendant-ancestor relationships 
to reduce repetitive computations. This results in
more efficient index construction and query processing.

Our contributions are summarized as follows:
\begin{itemize}
\item{We propose a new type of graph analytic query, \emph{Graph Window Query} and formally define two graph windows: k-hop window and topological window. We illustrate how these window queries would help users better query and
understand the graphs under these different semantics.}

\item{ To support efficient query processing, we further propose two different types of indices: \emph{Dense Block Index} (DBIndex) and \emph{Inheritance Index} (I-Index). The
\emph{DBIndex} and \emph{I-Index} are specially 
optimized to support k-hop window and topological window query processing. 
We develop the indices by integrating the window aggregation sharing techniques to salvage partial work done for efficient computation. In addition, we develop space and performance efficient techniques for index construction. 
Compared to EAGR \cite{mondal2014eagr}, the state-of-the-art index method for k-hop window queries, our DBIndex is much more memory efficient and scalable towards handling the large-scale graphs. }

\item{We perform extensive experiments over both real and synthetic datasets
with hundreds of millions of vertices and edges on a single machine. Our experiments 
indicate that our proposed index-based algorithms outperform the naive non-index
algorithm 
by up to four orders of magnitude. In addition, our experiments also show 
that DBIndex is superior over EAGR in terms of both
scalability and efficiency. In particular, 
DBIndx saves up to 80\% of indexing time as 
compared to EAGR, 
and performs well even when EAGR fails due
to memory limitations. 
}
\end{itemize}

%% file: sec2_related_work.tex
\section{Related Work}
Our proposed graph window functions (GWFs) for graph databases is inspired by the usefulness of window functions in relational analytic queries
\cite{zemke2012s}.

A window function in SQL typically specifies a set of partitioning attributes $A$ and an aggregation function $f$.
Its evaluation first partitions the input records based on $A$ to compute $f$ for each partition,
and each input record is then associated with the aggregate value corresponding to the partition that contains the record.
Several optimization techniques \cite{cao2012optimization, bellamkonda2013adaptive}
have also been developed to evaluate complex SQL queries involving multiple window functions.

However, the semantics and evaluation of window functions are very different between relational and graph contexts.
Specifically, the partitions (i.e., subgraphs) associated with GWFs are not necessarily disjoint; thus,
the evaluation techniques developed for relational context \cite{cao2012optimization, bellamkonda2013adaptive} are not applicable to GWFs. 

GWFs are also different from graph aggregation \cite{zhao2011graph,wang2014pagrol,chen2008graph,tian2008efficient} in graph OLAP.
In graph OLAP, information in a graph are summarized
by partitioning the graph's nodes/edges (based on some attribute values) and computing aggregate values for each partition.
GWFs, on the other hand, compute aggregate values for each graph node w.r.t. the subgraph associated with the node.
Indeed, such differences also arise in the relational context, where different 
techniques are developed to evaluate OLAP and window function queries.

In \cite{yan2010top}, the authors investigated the problem of finding the vertices that have top-k highest aggregate values over their h-hop neighbors. They proposed mechanisms to prune the computation by using two properties: First, the locality between vertices is used to propagate the upper-bound of aggregation; Second, the upper-bound value of aggregates can be estimated from the distribution of attribute values. However, all these pruning techniques are not applicable in our work, as we need to compute the aggregation value for every vertex. In such a scenario, techniques in \cite{yan2010top} degrade to the non-indexed approach as described in Section 4.


Indexing techniques have been proposed to efficiently determine whether an input pair of vertices is within a distance of k-hops (e.g. k-reach index \cite{cheng2012k}) or reachable (e.g. reachability index \cite{yu2010graph}). However, such techniques are not efficient for computing the k-hop window or topological window for a set of $n$ vertices with a time complexity of $O(n^2)$.

\cite{mondal2014eagr} proposed an EAGR system, which uses the famous VNM heuristic and Frequent-Pattern
Tree to find the shared component among each vertex's neighborhoods. It starts by building an
overlay graph as a bipartite graph representing the vertex-neighbor mapping. Then it aims to
find the bi-cliques in the overlay graph. Each bi-clique represents a set of vertices whose neighborhood
aggregates can be shared. Once a bi-clique is found, it is inserted back to the overlay graph as
a virtual node to remove redundant edges. EAGR find bicliques in iterations. During each iteration,
it sorts each vertices in overlay graph by their neighborhood information. Then the sorted vertices 
are split into equal-sized chunks. For each chunk, it then builds a FP-Tree to mining the large bi-cliques. 
As the algorithm iterates, the overlay graph evolves to be less dense.

The main drawback of EAGR is its demands of high memory usage on overlay construction. It
requires the neighborhood information to be pre-computed, which is used in the sorting phase
of each iteration. In EAGR the neighborhood information is assumed to be stored in memory.
However, the assumption does not scale well for computing higher hop windows (such as k $\geq 2$). 
For instance, a LiveJournal social network graph \footnote{Available at http://snap.stanford.edu/data/index.html, which is used \cite{mondal2014eagr}} (4.8M nodes, 69M edges) generates over 100GB mapping information for k=2 in adjacency list 
representation. If the neighborhood information is resided in disk, the performance of EAGR will
largely reduced. Similarly, if the neighborhood information is computed on-the-fly, EAGR needs
to perform the computation in each iteration, which largely increases indexing time.

We tackle this drawback by adapting a hash based approach that clusters each vertex based
on its neighborhood similarity. During the hashing, the vertex's neighborhood information
is computed on-the-fly. As compared to sorting based approach, we do not require vertex's neighborhood
to be pre-reside in memory. In order to reduce the repetitive computation of vertex's neighborhood,
we further propose an estimation based indexing construction algorithm that only require a vertex's small
hop neighborhood to be computed during clustering. As our experiments show, our proposed methods 
can perform well even when EAGR algorithm fails when neighborhood information overwhelms system's memory.
To further reduce the neighborhood access, we adapted a Dense Block heuristic process each vertex
in one pass. Experiments shows that the performance of our heuristic is comparable to EAGR's, but with
much shorter indexing time.

%% file: sec3_definition.tex
\section{Problem Formulation}

In this section, we provide the formal definition of graph window query.
We use $G = (V,E)$ to denote a directed/undirected data graph, where $V$ is its vertex set and $E$ is its edge set.
Each node/edge is associated with a (possibly empty) set of attribute-value pairs.


Fig.~\ref{fig:attributed} shows an undirected graph representing a social network that we will use as our running example in this paper. 
The table shows the values of the five attributes (User, Age, Gender, Industry, and Number of posts) associated with each vertex. For convenience, each node is labeled with its user attribute value;
and there is one edge between a user X and another user Y if X and Y are connected in the social network.

\begin{figure}[t]
\centering
\includegraphics[width=0.48\textwidth]{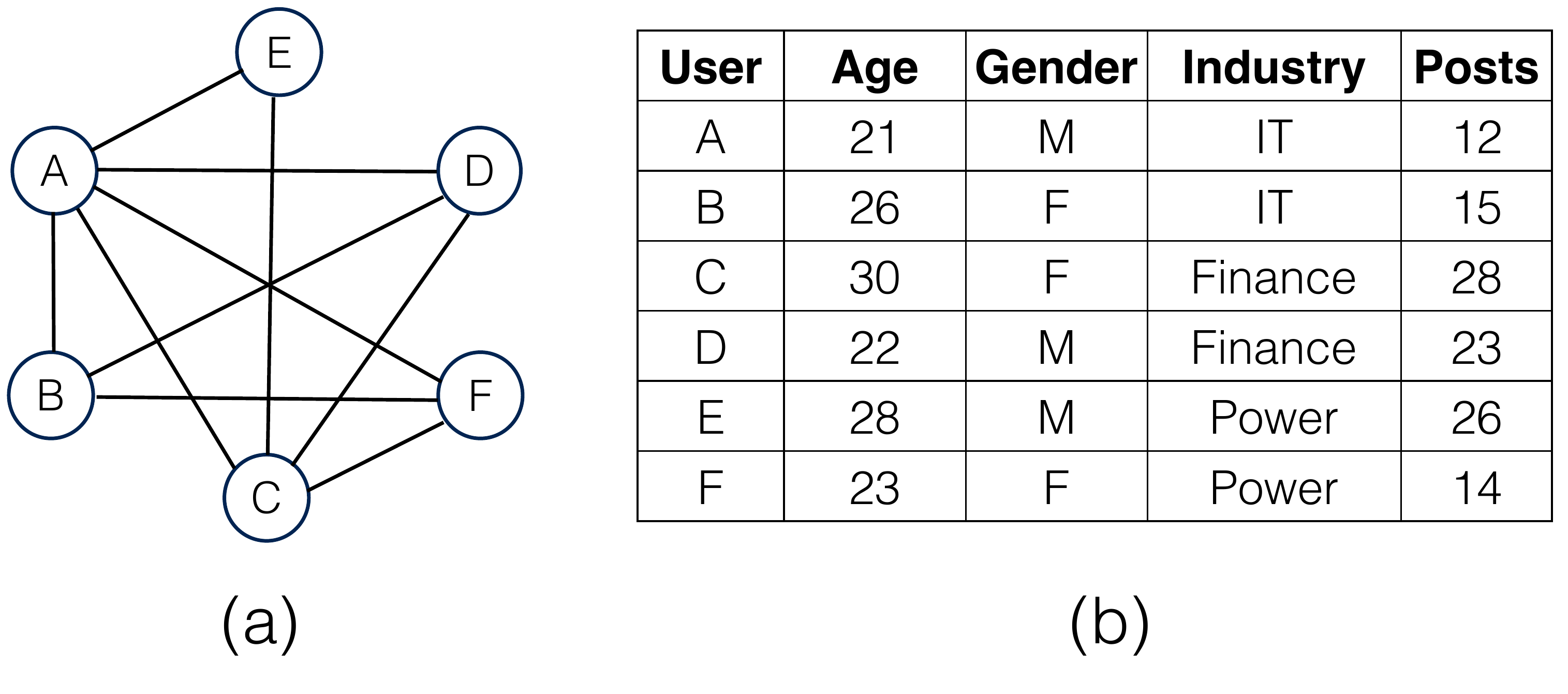}
	\caption{Running Example of Social Graph. (a) provides the graph structure; (b) provides the attributes associated with the vertices of (a).} 
	\label{fig:attributed}
\end{figure}

Given a data graph $G = (V,E)$,
a \emph{Graph Window Function (GWF)} over $G$ can be expressed 
as a quadruple $(G, W, \Sigma, A)$, where 
$W(v)$ denotes a \emph{window specification} for a vertex $v \in V$ 
that determines the set of vertices in some subgraph of $G$,
$\Sigma$ denotes an \emph{aggregation function}, and $A$ denotes 
a \emph{vertex attribute}.
The evaluation of a GWF $(G, W, \Sigma, A)$ on $G$
computes for each vertex $v$ in $G$, the aggregation $\Sigma$ on the 
values of attribute $A$  
over all the vertices in $W(v)$, which we denote by $\Sigma_{v' in W(v)} v'.A$.

Note that, in this paper, we focus on the attribute-based aggregation with distributive or algebraic aggregation functions. 
In other words, $W(v)$ refers to a set of vertices, and the aggregation function 
$\Sigma$ operates on the values of attribute $A$ over all the vertices in $W(v)$. Meanwhile, the aggregation function $\Sigma$ is distributive or algebraic (e.g., sum, count, average), as these aggregation functions are widely used in practice. 
 
In the following, we introduce two useful types of window specification (i.e., $W$), namely, 
\emph{k-hop window} and \emph{topological window}.

\begin{defn}[K-hop Window] 
Given a vertex $v$ in a data graph $G$, 
the $k$-hop window of $v$, denoted by $W_{kh}(v)$ (or $W(v)$ when there is no ambiguity),
is the set of neighbors of $v$ in $G$ which can be reached within $k$ hops.
For an undirected graph $G$,
a vertex $u$ is in $W_{kh}(v)$  iff there is a $\alpha$-hop path between $u$ and $v$ where $\alpha \leqslant k$.
For a directed graph $G$,
a vertex $u$ is in $W_{kh}(v)$  iff there is a $\alpha$-hop directed path from $v$ to $u$ \footnote{
Other variants of k-hop window for directed graphs are possible; e.g.,
a vertex $u$ is in $W_{kh}(v)$  iff there is a $\alpha$-hop directed path from $u$ to $v$ where $\alpha \leqslant k$.
} where $\alpha \leqslant k$.
\end{defn}

Intuitively, a k-hop window selects the neighboring vertices of a 
vertex within a  k-hop distance. 
These neighboring vertices typically represent the most important 
vertices to a vertex with regard to their structural relationship in a graph. 
Thus, k-hop windows provide meaningful specifications for many applications, such as customer behavior analysis \cite{Briscoe2013Credit,Dai2012Predict} , digital marketing \cite{Ma2009Marketing} etc.

As an example, in Fig.~\ref{fig:attributed}, the $1$-hop window of vertex \emph{E} is $\{A,C,E\}$ and the $2$-hop window of vertex \emph{E} is $\{A,B,C,D,E,F\}$.  

\begin{defn}[Topological Window] 
Given a vertex $v$ in a DAG $G$, the topological window of $v$, denoted by $W_t(v)$,
refers to the set of ancestor vertices  of $v$ in $G$;
i.e., a vertex $u$ is in $W_t(v)$ iff there is directed path from $u$ to $v$ in $G$.
\end{defn}

There are many directed acyclic graphs (DAGs) in real-world applications (such as biological networks, citation networks and dependency networks)
where topological windows represent meaningful relationships that are of interest.
For example, in a citation network where (X,Y) is an edge  iff paper $X$ cites paper $Y$, 
the topological window of a paper represents the citation impact of that paper \cite{Ma2008Pagerank,Holsapple2003CAI,Campanario:2011ESJ}.

As an example, Fig.~\ref{fig:topological} shows a small example of a Pathway Graph from a biological network. 
The topological window of $E$ $W_t(E)$ is $\{A, B, C, D, E\}$ and $W_t(H)$ is $\{A, B, D, H\}$.

\begin{figure}[h]
\centering
 \includegraphics[width=45mm,height=45mm]{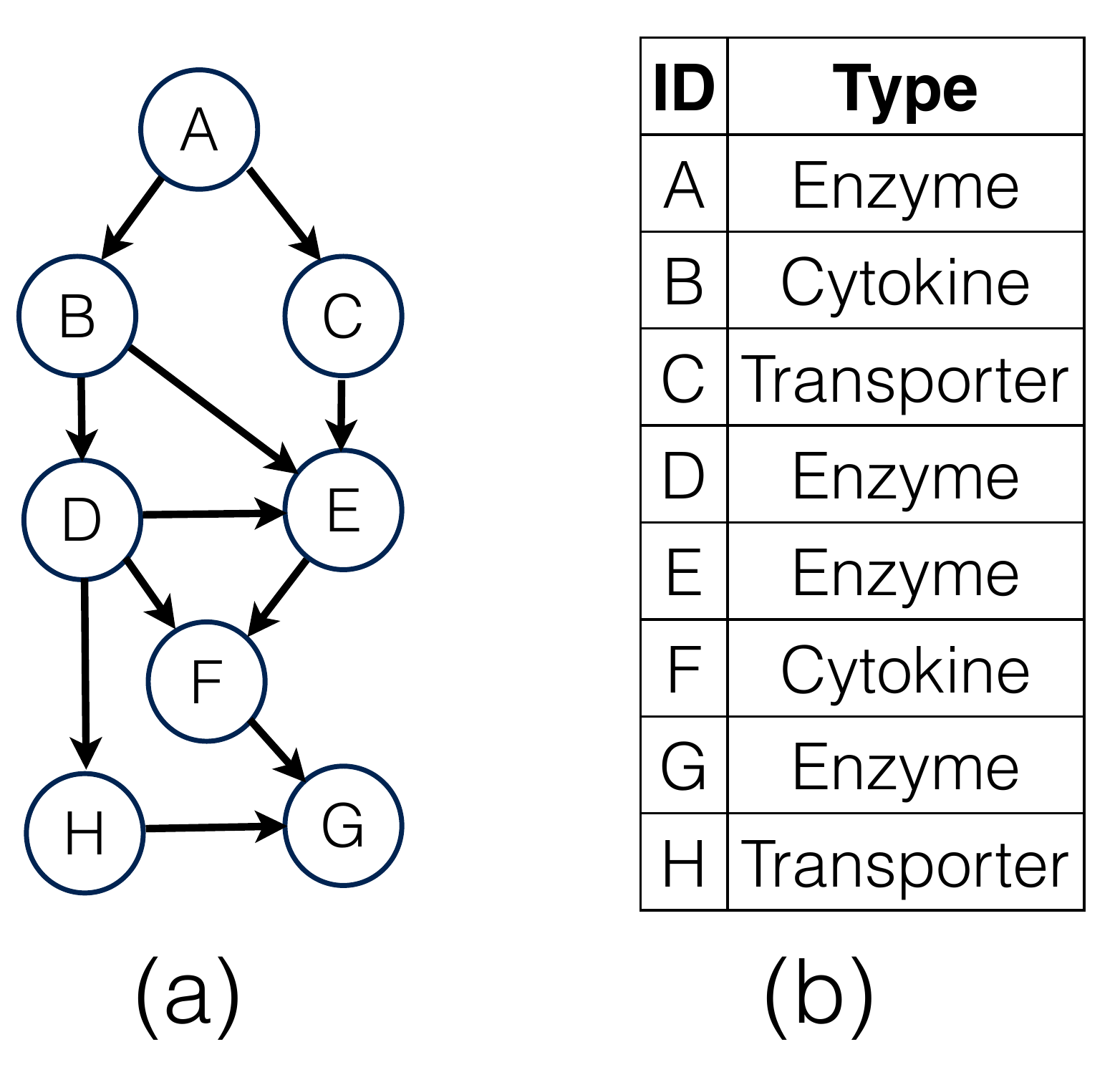}
	\caption{Running Example of Pathway DAG. (a) provides the DAG structure; (b) provides the attributes associated with the vertices of (a).}
	\label{fig:topological}
\end{figure}

\begin{defn}[Graph Window Query] 
A graph window query on a data graph $G$ is of the form
$GWQ(G, W_1, \Sigma_1, A_1,\cdots,$ \\
$W_m, \Sigma_m, A_m)$, where $m \geq 1$
and
each quadruple $(G, W_i,\Sigma_i,A_i)$ is a graph window function on $G$.
\end{defn}
In this paper, we focus on graph window queries with a single window 
function that is either a k-hop or topological window. 
The evaluation of complex graph window queries with multiple window 
functions can be naively processed as a sequence of window functions one
after another. We leave the optimization of processing multiple
window functions for further study in the future.

%% file: sec4_DBIndex.tex
\newcommand{\DBIndex}{DBIndex}
\newcommand{\blockset}{{\cal B}}
\newcommand{\clusterset}{{\cal C}}

\section{Dense Block Index}
A straightforward approach to process a graph window query 
$Q = (G, W, \Sigma, A)$, where $G = (V,E)$,
is to dynamically compute the window $W(v)$ for each vertex $v \in V$ and
its aggregation 
$\Sigma_{v' \in W(v)} v'.A$ 
independently from other vertices. We refer to this approach as \emph{Non-Indexed} method.

Given that many of the windows would share many common nodes (e.g., the k-hop windows of two adjacent nodes),
such a simple approach would be very inefficient due to the lack of sharing of the aggregation computations. 

To efficiently evaluate graph window queries, we propose an indexing technique named \emph{dense block index} (\textit{\DBIndex}), which is both space and query efficient. 
The main idea of \DBIndex\ is to try to reduce the aggregation computation cost by identifying subsets of nodes that are shared by more than one window 
so that the aggregation for the shared nodes could be computed only once instead of multiple times.

For example, consider a graph window query on the social graph in Fig.~\ref{fig:attributed} using the 1-hop window function.
We have $W(B) = \{A,B,D,F\}$ and $W(C) = \{A,C,D,E,F\}$ sharing three common nodes $A$, $D$, and $F$.
By identifying the set of common nodes $S=\{A,D,F\}$, its aggregation 
$\Sigma_{v \in S} v.A$ can be computed only once
and then reuse to compute the aggregation for $\Sigma_{v \in W(B)} v.A$ and $\Sigma_{v \in W(C)} v.A$.

Given a window function $W$ and a graph $G=(V,E)$,
we refer to a non-empty subset $B \subseteq V$ as a {\it block}.
Moreover, if $B$ contains at least two nodes and $B$ is contained by at least two different windows
(i.e., there exists $v_1, v_2 \in V$, $v_1 \neq v_2$, $B \subseteq W(v_1)$, and $B \subseteq W(v_2)$),
then $B$ is referred to as a {\it dense block}.
Thus, in the last example, $\{A,D,F\}$ is a dense block.

We say that a window $W(X)$ is {\it covered} by a collection of disjoint blocks $\{B_1,\cdots,B_n\}$
if the set of nodes in the window $W(X)$ is equal to the union of all nodes in the collection of disjoint blocks;
i.e., $W(X) = \bigcup_{i=1}^{n} B_i$ and $B_i \cap B_j = \emptyset$ if $i \neq j$.

To maximize the sharing of aggregation computations for a graph window query, 
the objective of \DBIndex\ is to identify a small set of blocks $\blockset$ such that
for each $v \in V$, $W(v)$ is covered by a small subset of disjoint blocks in $\blockset$.
Clearly, the cardinality of $\blockset$ is minimized if $\blockset$ contains a few large dense blocks.

Thus, given a window function $W$ and a graph $G=(V,E)$,
a \DBIndex\ to evaluate $W$ on $G$ consists of three components in the form of a bipartite graph.
The first component is a collection of nodes (i.e., $V$);
the second component is a collection of blocks; i.e., $\blockset = \{B_1,\cdots,B_n\}$ where each $B_i \subseteq V$;
and the third component is a collection of links from blocks to nodes
such that if a set of blocks $B(v) \subseteq \blockset$ is linked to a node $v \in V$,
then $W(v)$ is covered by $B(v)$.
Note that a \DBIndex\ is independent of both the aggregation function (i.e., $\Sigma$) and the attribute to be aggregated (i.e., $A$).

Fig.~\ref{fig:dbi_agg}(a) shows an example of a \DBIndex\ wrt the social graph in Fig.~\ref{fig:attributed} and the 1-hop window function.
Note that the index consists of a total of seven blocks of which three of them are dense blocks.

\begin{figure}[t]
\centerline{\epsfig{figure=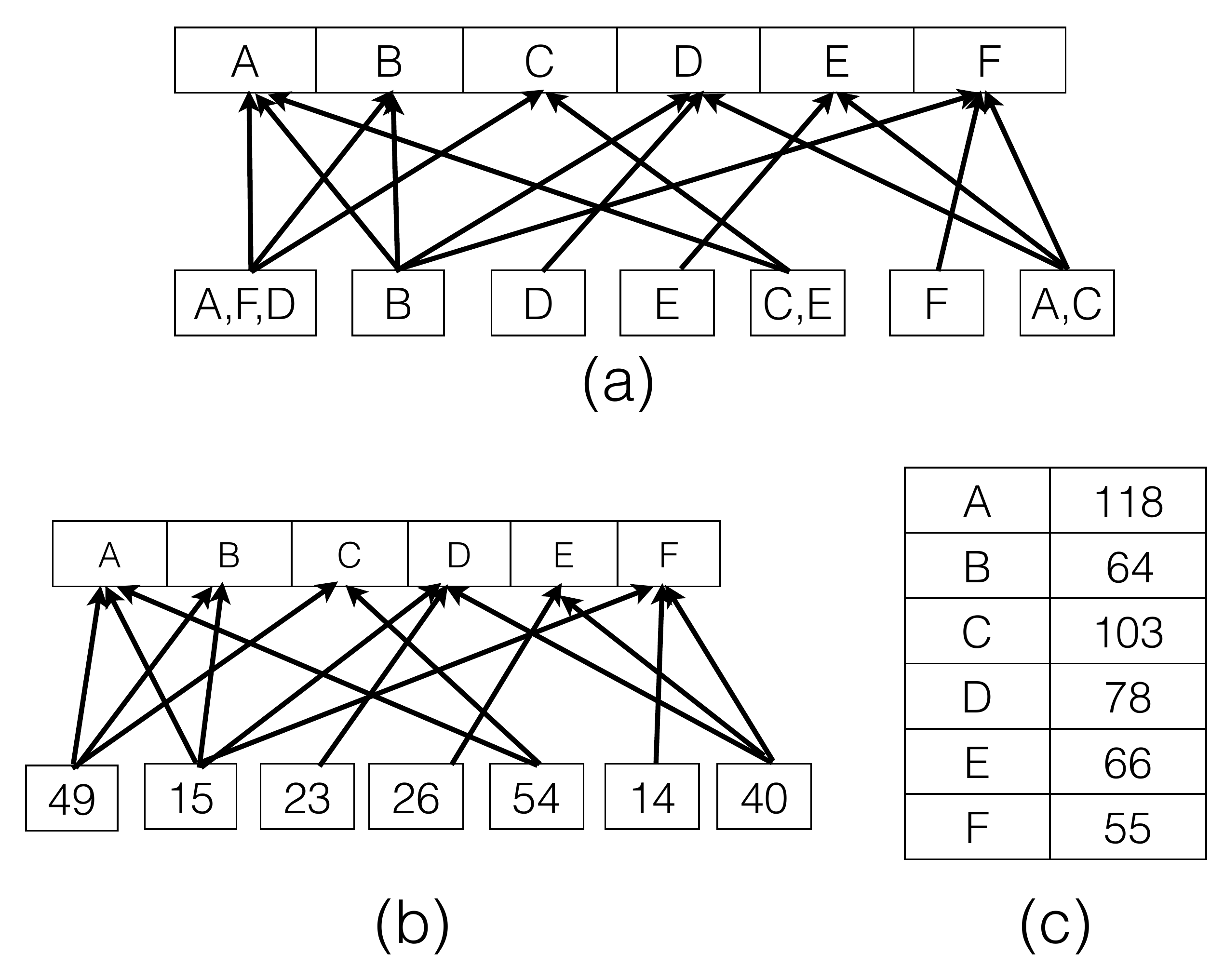,width=62.7mm} }
	\caption{Window Query Processing using DBIndex. (a) provides the DBIndex for 1-hop window query in Fig.~\ref{fig:attributed}; (b) shows the partial aggregate results based on the dense block; (c) provides the final aggregate value of each window. }
	\label{fig:dbi_agg}
\end{figure}

\subsection{Query Processing using \DBIndex}
Given a \DBIndex\ wrt a graph $G$ and a window function $W$, a graph window query $Q = (G, W, \Sigma, A)$ is processed by the following two steps.
First, for each block $B_i$ in the index, we compute the aggregation (denoted by $T_i$) over all the nodes in $B_i$;
i.e., $T_i = \Sigma_{v \in B_i} v.A$. 
Thus, each $T_i$ is a partial aggregate value.
Next, for each window $W(v)$, $v \in V$, the aggregation for the window is computed by aggregating over all the partial aggregates
associated with the blocks linked to $W(v)$;
i.e., if $B(v)$ is the collection of blocks linked to $W(v)$, 
then the aggregation for $W(v)$ is given by $\Sigma_{B_i \in B(v)} T_i$. 

Consider again the \DBIndex\ shown in Fig.~\ref{fig:dbi_agg}(a) 
defined wrt the social graph in Fig.~\ref{fig:attributed} and the 1-hop window function.
Fig.~\ref{fig:dbi_agg}(b) shows how the index is used to evaluate the graph window query $(G, W, sum, Posts)$
where each block is labeled with its partial aggregate value;
and
Fig.~\ref{fig:dbi_agg}(c) shows the final query results.


\subsection{\DBIndex\ Construction} 

In this section, we discuss the construction of the \DBIndex\ (wrt a graph $G=(V,E)$ and window function $W$) which has two key challenges.

The first challenge is the time complexity of the index construction. 
From our discussion of query processing using DBIndex, we note that the number of aggregation computations is determined by both the number of blocks as well as the number of links in the index; 
the former determines the number of partial aggregates to compute
while the latter determines the number of aggregations of the partial aggregate values.
Thus, to maximize the shared aggregation computations using DBIndex, both the number of blocks in the index as well as the number of blocks covering each window should be minimized. 
However, finding the optimal \DBIndex\ to minimize this objective is NP-hard\footnote{
Note that a simpler variation of our optimization problem has been proven to be NP-hard \cite{vassilevska2004finding}.}.
Therefore, efficient heuristics are needed to construct the \DBIndex.

The second challenge is the space complexity of the index construction.
In order to identify large dense blocks to optimize for query efficiency,
a straightforward approach  is to first derive the window $W(v)$ for each node $v \in V$ and
then use this derived information to identify large dense blocks.
However, this direct approach incurs a high space complexity of $O(|V|^2)$.
Therefore, a more space-efficient approach is needed in order to scale to handle large graphs.

In this section, we present two heuristic approaches, namely {\it MC} and {\it EMC}, to construct \DBIndex.
The second approach EMC is designed to improve the efficiency of the first approach MC for constructing \DBIndex\
(wrt k-hop window function) by using some approxmation techique at the expense of possibly sacrificing 
the ``quality'' of the dense blocks (in terms of their sizes).

\subsubsection{MinHash Clustering (MC)}

To reduce both the time and space complexities for the index construction,
instead of trying to identify large dense blocks among a large collection of windows,
we first partition all the windows 
into a number of smaller clusters of similar windows and then identify large dense blocks for each of the smaller clusters.
Intuitively, two windows are considered to be highly similar if they share a larger subset of nodes.
We apply the well-known {\it MinHash based Clustering (MC)} algorithm \cite{broder1997syntactic} to partition the windows into clusters of similar windows.
The MinHash clustering algorithm is based on using the
\emph{Jaccard Coefficient} to measure the similarity of two sets.  
Given the two window $W(v)$ and $W(u)$, $u,v \in V$,
their \emph{Jaccard Coefficient} is given by
\begin{equation} \label{eq:jacc_sim}
	J(u,v) = \frac{|W(u) \cap W(v)|}{|W(u) \cup W(v)|}
\end{equation}
The \emph{Jaccard Coefficient} ranges from 0 to 1, where  a larger value means that the windows are more similar.

\begin{figure}[htb]
\centering
\includegraphics[width=65mm,height=55mm]{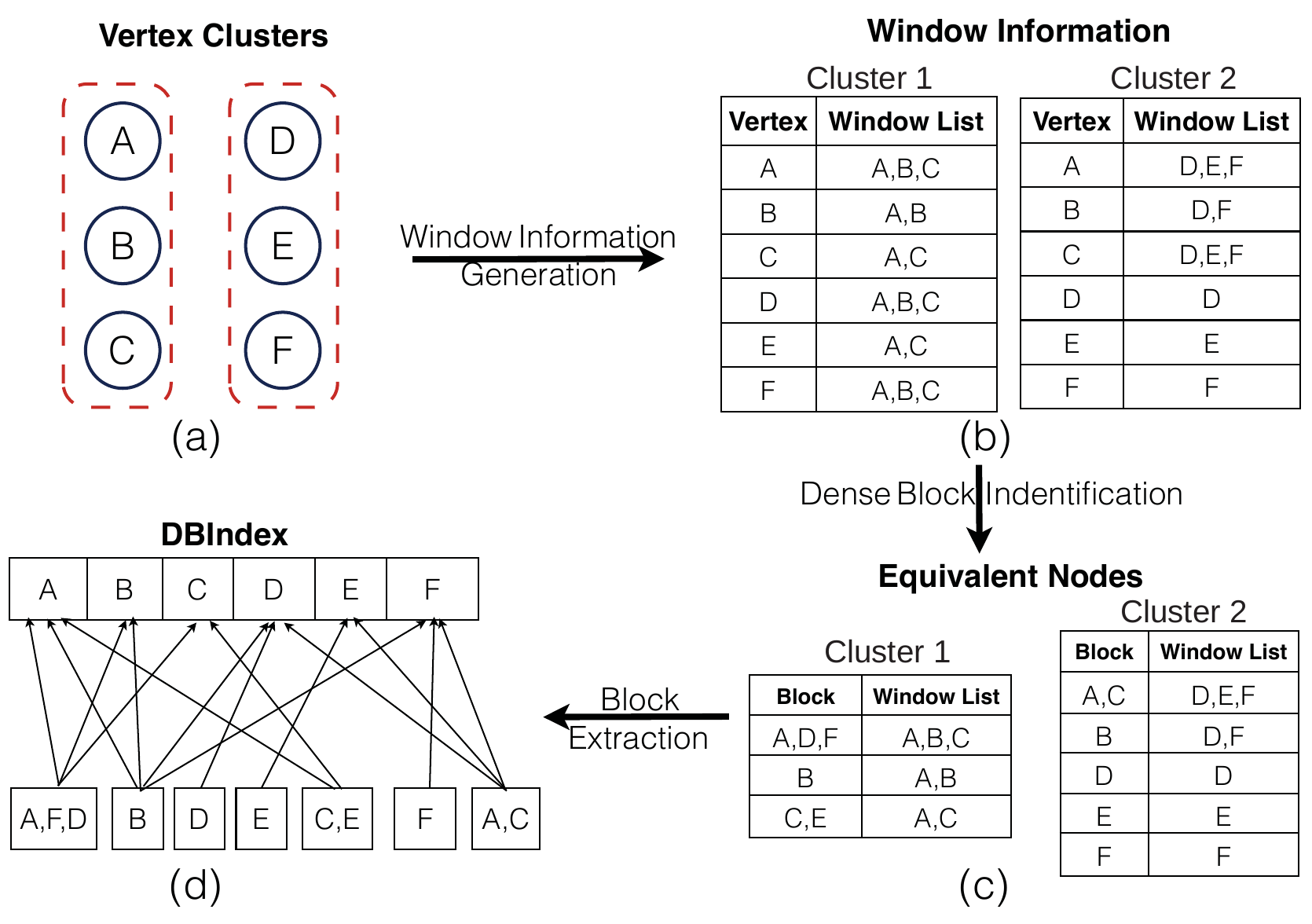}
\caption{DBIndex Construction over Social Graph in Fig.~\ref{fig:attributed}. (a) shows two clusters after MinHash clustering; (b) shows the window information of involved vertices within each cluster; (c) shows the dense blocks within each cluster; (d) provides the final DBIndex.}
\label{fig:dbi-indexing}
\end{figure}

Our heuristic approach to construct \DBIndex\ $I$ operates as follows.
Let $nodes(I)$, $blocks(I)$, and $links(I)$ denote, respectively,
the collection of nodes, blocks, and links that form $I$.
Initially, we have $nodes(I) = V$,
$blocks(I) = \emptyset$,
and
$links(I) = \emptyset$.
\begin{algorithm}
\caption{CreateDBIndex}
\begin{algorithmic}[1]  \small
\REQUIRE Graph $G=(V,E)$, window function $W$
\ENSURE \DBIndex\ $I$
\STATE Initialize \DBIndex\ $I$: $nodes(I)=V$, $blocks(I)=\emptyset$, $links(I)=\emptyset$
\FORALL{$v \in V$}
	\STATE Traverse $G$ to determine $W(v)$
	\STATE Compute the hash signature $H(v)$ for $W(v)$
\ENDFOR
\STATE Partition $V$ into clusters $\clusterset = \{C_1,C_2,\cdots\}$ based on hash signatures $H(v)$
\FORALL{$C_i \in \clusterset$}
	\FORALL{$v \in C_i$} 
		\STATE Traverse $G$ to determine $W(v)$
	\ENDFOR
	\STATE {\tt IdentifyDenseBlocks} $(I,G,W,C_i)$
\ENDFOR
\RETURN $I$
\end{algorithmic}
\label{algo:k-hop-dbi}
\end{algorithm}

The first step is to partition the nodes in $V$ into clusters
using MinHash algorithm such that nodes with similar windows belong to the same cluster. 
For each node $v \in V$, we first derive its window $W(v)$ by an appropriate traversal of the graph $G$.
Next, we compute a hash signature (denoted by $H(v)$) for $W(v)$ based on applying $m$ hash functions on the set $W(v)$.
Nodes with identical hash signatures are considered to have highly similar windows and are grouped into the same cluster.
To ensure that our approach is scalable,
we do not retain $W(v)$ in memory  after its hash signature $H(v)$ has been computed and used to cluster $v$;
i.e., our approach does not materialize all the windows in the memory to avoid high space complexity.
Let $\clusterset = \{C_1,C_2,\cdots\}$ denote the collection of clusters obtained from the first step,
where each $C_i$ is a subset of nodes.

The second step is to identify dense blocks from each of the clusters computed in the first step.
The identification of dense blocks in each cluster $C_i$ is based on the notion of node equivalence defined as follows.
Two distinct nodes $u, v \in C_i$ are defined to be equivalent (denoted by $u \equiv v$)
if $u$ and $v$ are both contained in the same set of windows;
i.e., for every window $W(x), x \in C_i$, $u \in W(x)$ iff $v \in W(x)$.
Based on this notion of node equivalence, $C_i$ is partitioned into blocks of equivalent nodes.
To perform this partitioning, we need to again traverse the graph for each node $v \in C_i$ to 
determine its window $W(v)$\footnote{
Note that although we could have avoided deriving $W(v)$ a second time if we had materialzed all the derived windows the first time, our approach is designed to avoid the space complexity of materializing all the windows in memory at the cost of computing each $W(v)$ twice. We present an optimization in Section~\ref{sec:optimized} to avoid the recomputation cost.
}.

However, since $C_i$ is now a smaller cluster of nodes, we can now materialize all the windows for the nodes in $C_i$ in memory without exceeding the memory space. In the event that a cluster $C_i$ is still too large for all its vertex windows to be materialized in main memory, we can further partition $C_i$ into equal size sub-clusters. This re-partition process can be recursively performed until the sub clusters created are small enough such that the windows for all nodes in the sub clusters fit in memory.

Recall that a block $B$ is a dense block if $B$ contains at least two nodes and
$B$ is contained in at least two windows.
Thus, we can classify the nodes in each $C_i$ as either dense or non-dense nodes:
a node $v \in C_i$ is classified as a {\it dense node} if $v$ is contained in a dense block;
otherwise, $v$ is a non-dense node.

For each dense block $B$ in $C_i$,
we update the blocks and links in the \DBIndex\ $I$ as follows:
we insert $B$ into $blocks(I)$ if $B \not\in blocks(I)$,
and we insert $(B,v)$ into $links(I)$
for each $v \in C_i$ where $B \subseteq W(v)$.
If all the blocks in $C_i$ are dense blocks, then we are done with identifying dense blocks in $C_i$;
otherwise, there are two cases to consider.
For the first case, if all the nodes in $C_i$ are non-dense nodes,
then we also terminate the process  of identifying dense blocks in $C_i$
and update the blocks and links in the \DBIndex\ $I$ as before:
we insert each non-dense block $B$ into $blocks(I)$,
and we insert $(B,v)$ into $links(I)$ for each $v \in C_i$ where $B \subseteq W(v)$.
For the second case, if $C_i$ has a mixture of dense and non-dense nodes,
we remove the dense nodes from $C_i$ and recursively identify dense blocks in $C_i$
following the above two-step procedure.

Note that since the blocks are identified independently from each cluster,
it might be possible for the same block to be identified from different clusters.
We avoid duplicating the same block in $blocks(I)$ by checking that a block $B$ is not already in $blocks(I)$
before inserting it into $blocks(I)$.
The details of the construction algorithm are shown as Algorithms
\ref{algo:k-hop-dbi},
\ref{algo:identify},
and
\ref{algo:refine}.


\begin{algorithm}
\caption{IdentifyDenseBlocks}
\begin{algorithmic}[1] \small
\REQUIRE \DBIndex\ $I$, Graph $G=(V,E)$, window function $W$, a cluster $C_i \subseteq V$
\STATE Partition $C_i$ into blocks based on node equivalence
\STATE Initialize $DenseNodes = \emptyset$
\FORALL{dense block $B$}
	\STATE Insert $B$ into $blocks(I)$ if $B \not\in blocks(I)$
	\STATE Insert $(B,v)$ into $links(I)$ for each $v \in C_i$ where $B \subseteq W(v)$
	\STATE $DenseNodes = DenseNodes \cup B$
\ENDFOR
\IF{($DenseNodes = \emptyset$)}
	\FORALL{block $B$}
		\STATE Insert $B$ into $blocks(I)$ if $B \not\in blocks(I)$
		\STATE Insert $(B,v)$ into $links(I)$ for each $v \in C_i$ where $B \subseteq W(v)$
	\ENDFOR
\ELSIF{($C_i - DenseNodes \neq \emptyset$)}
	\IF{($C_i \neq DenseNodes$)}
	\STATE {\tt RefineCluster} $(I,G,W,C_i - DenseNodes)$
	  \ENDIF
\ENDIF
\end{algorithmic}
\label{algo:identify}
\end{algorithm}

Fig.~\ref{fig:dbi-indexing} 
illustrates the construction of the DBIndex with respect to the social graph in 
Fig.~\ref{fig:attributed}(a) and 1-hop window using the MC algorithm.
First, the set of graph vertices are partitioned into clusters using MinHash clustering;
Fig.~\ref{fig:dbi-indexing}(a)
shows that the set of vertices $V = \{A, B, C, D, E, F \}$ are partitioned into two clusters $C_1=\{A, B, C\}$ and $C_2=\{D, E, F\}$. 

For convenience, cluster 1 in 
Fig.~\ref{fig:dbi-indexing}(b) shows for each $v \in C_1$, the set of vertices whose windows contain $v$;
i.e., $\{u |\ v \in W(u)\}$.
Similarly, Cluster 2 in Fig.~\ref{fig:dbi-indexing} (b)
shows for each $v \in C_2$, the set of vertices whose windows contain $v$.
Consider the identification of dense blocks in cluster $C_1$.
As shown in Fig.~\ref{fig:dbi-indexing} (c), based on the notion of equivalence nodes,
cluster $C_1$ is partitioned into three blocks of equivalent nodes:
$B_1=\{A,D,F\}$, $B_2=\{B\}$, and $B_3\{C,E\}$.
Among these three blocks, only
$B_1$ and $B_3$ are dense blocks.
The MC algorithm then tries to repartition the window $A,B,C$ using non-dense nodes in $C_1$,
(i.e., $B_2$). Since $B_2$ is the only non-dense node, it directly outputs.
At the end of processing cluster $C_1$,
the DBIndex $I$ is updated as follows:
$blocks(I) = \{B_1, B_2, B_3\}$ 
and
$links(I) = \{ (B_1,\{A,B,C\}), (B_2, \{A,B\}),$, $(B_3, \{A,C\}) \}$. The identification of dense blocks in cluster $C_2$ 
is of similar process.

We find it is non-trivial to precisely analyze the complexity of Algorithm~\ref{algo:k-hop-dbi}. Here, we only offer a brief analysis. Suppose the MinHash cost is $H$ and the total cost for k-bounded BFS for all vertex is $B$, Lines 1-5 has the complexity of $O(H + B)$.  Lines 7-10 has the complexity of $O(B)$. A single execution of Algorithm~\ref{algo:identify}  has the  complexity of $O(|V|)$, since we can simply partition nodes using hashing. Suppose the iteration runs for $K$ times, the total cost for Algorithm~\ref{algo:identify} and Algorithm~\ref{algo:refine} is $O(K|V|)$. Therefore the overall complexity of Algorithm~\ref{algo:k-hop-dbi} is $O(H+2*B + O(K|V|))$. $H$ depends on the number of vertex-window mappings for a given query and $B$ depends on the graph structure and number of hops. As we demonstrate in Section 6, the $H$ and $B$ are the major contribution of the indexing time. To reduce the index time, we provide further optimization techniques.

\begin{algorithm}
\caption{RefineCluster}
\begin{algorithmic}[1]
\REQUIRE \DBIndex\ $I$, Graph $G=(V,E)$, window function $W$, a cluster $C \subseteq V$
\FORALL{$v \in C$}
	\STATE Compute the hash signature $H(v)$ for $W(v) \cap C$
\ENDFOR
\STATE Partition $C$ into clusters $\clusterset = \{C_1,C_2,\cdots\}$ based on hash signatures $H(v)$
\FORALL{$C_i \in \clusterset$}
	\STATE {\tt IdentifyDenseBlocks} $(I,G,W,C_i)$
\ENDFOR
\end{algorithmic}
\label{algo:refine}
\end{algorithm}

\subsubsection{Estimated MinHash Clustering (EMC)}
\label{sec:optimized}

The MC approach described in the previous section requires the window of each node (i.e., $W(v), v \in V$)
to be computed twice in order to avoid the high space complexity of materializing all the windows in main memory.
For k-hop window function with a large value of $k$, the cost of graph traverals to compute the k-hop windows
could incur a high computation overhead. Moreover, the cost of initial MinHash in MC approach equals to the initial number of vertex-window mappings, which is of the same order as graph traversal. For the larger hops, MinHash clustering would incur high computation cost.

To address these issues, we present an even more efficient approach,
referred to as {\it Estimated MinHash Clustering (EMC)}, 
to optimize the construction of the \DBIndex\ for k-hop window function with larger k.

The key idea behind $EMC$ is based on the observation that for any two nodes $u, v \in V$,
if their $m$-hop windows, $W_m(u)$ and $W_m(v)$, are highly similar 
and they are grouped into the same cluster, 
then it is likely that the $n$-hop windows of these two nodes, where $n > m$,
would also be highly similar and grouped into the same cluster.

Using the above observation, we could reduce the overhead cost for constructing a \DBIndex\ wrt a $k$-hop window 
function by clustering the nodes based on their $k'$-hop windows, where $k' < k$, instead of their $k$-hop windows.

To reduce the overhead of window computations,
our $EMC$ approach is similar to the MC approach except   
for the first round of window computations
(line 3 in Algorithm \ref{algo:k-hop-dbi}):
$EMC$ uses lower hop windows to approximate $k$-hop windows for the purpose of clustering the nodes in $V$.
Thus, the hash signatures used for partitioning $V$ are based on lower hop windows.
This approximation clearly has the advantage of improved time-efficiency as traversing and Minhashing on lower hop window is of
order of magnitude faster. For the extreme case, adapting 1-hop window of a node $v$ requires only accessing the adjacent nodes of $v$. 
The tradeoff for this improved efficiency is that the ``quality'' of the dense blocks might be reduced (in terms of their sizes).
However, our experimental results show that this reduction in quality is actually only marginal which makes this approximation a worthy tradeoff.

\subsubsection {Justification of Heuristic}
In the following, we show the theoretical justification of our heuristic: the Jaccard coefficient
is increasing wrt number of hops for a large class of graphs. We 
assume that the degree of vertices follows the same distribution, which is true in most
real-networks and random graph models\footnote{E.g. Social network, Preferential Attachment model etc. follow
power-law distribution. However, our analysis do not restrict on power-law distribution.}.  
This implies we can analyze vertices 
with their neighborhoods structure using a unified way. 

We use $d_i$ to indicate the degree of vertex $i$. For any vertex pair $(u,v)$, their 
intersection on $k$-hop window consists of three part. We name them using $A=W_k(u)-W_k(v)$, 
$B=W_k(v)-W_k(u)$ and $C_k = W_k(u) \cap W_k(v)$. Clearly the Jaccard coefficient at $k$-hops
can be expressed as follows:
\begin{equation}
	J_k(u,v) = \frac{|C_k|}{|A_k| + |B_k| + |C_k|}
\end{equation}

To deduct the relationship between $J_k(u,v)$ and $J_{k+1}(u,v)$, we prove the following lemma first:
\begin{thm}
Let $S$ be a collection of connected vertices, the number of newly discovered
vertices by one-hop expansion from $S$ is bounded by a function on $|S|$.
\end{thm}
\begin{proof}
Consider a random variable $Y_i$ indicate the newly 
discovered vertices from one-hop expansion from vertex $i$. 
Then the probability of $|Y_i|=y$ is can be analyzed as follows: there 
are $d_i$ edges for vertex $i$. Since $|Y_i|$ is connected with $S$, one edge
is fixed to link with a vertex in $S$. There are remaining $d_i-1$ edges with
$y$ edges linked to the new vertices. In total, there are $|V|-1 \choose d_i -1$
combinations with $d_i$ edges. Therefore, the probability can be written as:
\begin{equation}
Prob(Y_i = y| v_i \in S) = \frac{{|S| -1 \choose d_i - y -1}{|V|-|S| \choose y}}{{|V|-1 \choose d_i -1}}
\end{equation}
Thus, the expectation of $Y_i$ is:
\begin{equation}
\begin{split}
E(Y_i|v_i \in S) & = \Sigma( y * Prob(Y_i = y| v_i \in S) )\\
	& = \Sigma_{y=1}^{y=d_i -1} ( \frac{{|S| -1 \choose d_i - y -1}{|V|-|S| \choose y}}{{|V|-1 \choose d_i -1}} * y )\\
	& = \Sigma_{y=1}^{y=d_i -1} (\frac{{|S| -1 \choose d_i - y -1}{|V|-|S| -1 \choose y - 1}}{{|V|-1 \choose d_i -1}} * (|V|-|S|))\\
	& = (|V|-|S|) * \Sigma_{y=1}^{y=d_i -1}\frac{{|S| -1 \choose d_i - y -1}{|V|-|S| -1 \choose y - 1}}{{|V|-1 \choose d_i -1}} \\
	& = (|V|-|S|) * \frac{{|V|-2 \choose d_i - 2}}{{|V|-1 \choose d_i -1}} = \frac{(|V|-|S|)*(d_i-1)}{|V| - 1}
\end{split}
\end{equation}
Taking the expectation over all vertices in $S$, we can find the expectation of $E(Y|S) = \frac{(|V|-|S|)*(\overline{d}-1)}{|V| - 1}$, 
where $\overline{d}$ is the average degree of the graph. We then define the event $X=\cup_{i=1}^{i=|S|} Y_i$, i.e. $X$ is the number
of newly discovered vertices for one-hop expansion of entire $S$. By union bound, 
the expectation of $X$ is:
\begin{equation} 
\begin{split}
E(X|S) &= E(\cup_{i=1}^{i=|S|}Y_i|S) \\
& \leq \Sigma_{i=1}^{i=|S|}E(Y_i|S) \\ 
&= \frac{|S|(|V|-|S|)(\overline{d} -1)}{|V|-1} = f(|S|)
\end{split}
\end{equation}
The bound is achieved when each vertices in $S$ discovers non-overlapping neighbors, such
as in the case of tree structure. Therefore, the newly discovered vertices are tightly bounded 
by a quadratic function on $|S|$.
\qed
\end{proof}

We thus use $f(m)$ to denote the number of newly discovered vertices
from a base set of $m$ connected vertices. Since $u,v$ have 
identical degree distribution, their expected value $S_u=E(|W_k(u)|)$
and $S_v=E(|W_k(v)|)$ is likely to be the same, i.e. $S_u \cong S_v$.
We further use $\alpha = \frac{|C|}{|A|+|C|}$ to denote the portion of shared
components in $u$'s $k$-hop neighborhood. Likewise, we use $\beta = \frac{|B|}{|B|+|C|}$ for 
$W_k(v)$. Since vertices have identical degree distribution, $\alpha$ and $\beta$ are likely
to be the same, i.e. $\alpha \cong \beta$.
Now, the $J_{k+1}(u,v)$ for ($k+1$)-hop can be represented as follows:
\begin{equation}
\begin{split}
J_{k+1}(u,v) & = \frac{|C_{k+1}|}{|A_{k+1}| + |B_{k+1}| + |C_{k+1}|} \\
	& = \frac{|C_k| + \alpha * f(S_u) + \beta * f(S_v)}{|A_k| +|B_k| +|C_k| +f(S_u) +  f(S_v) -\Delta} 
\end{split}
\end{equation}
, the $\Delta$ here is to compensate the doubly counted portion on the
overlapping: $(A_{k+1} \cup B_{k+1}) \cap C_{k+1}$. Therefore, it subsequently follows:
\begin{equation}
\begin{split}
J_{k+1}(u,v) & \geq \frac{|C_k| + \alpha * f(S_u) + \beta * f(S_v)}{|A_k| +|B_k| +|C_k| +f(S_u) +  f(S_v)}  \\
		& \cong \frac{|C_k| + 2\alpha * f(S_u)}{|A_k|+|B_k|+|C_k| + 2 * f(S_u)}
\end{split}
\end{equation}
Due to the fact that $\alpha = \frac{|C_k|}{|C_k|+|A_k|} \geq \frac{|C_k|}{|A_k|+|B_k|+|C_k|}$, it follows:
\begin{equation}
\begin{split}
	\frac{\alpha f(S_u)}{f(S_u)} & \geq \frac{|C_k|}{|A_k|+|B_k|+|C_k|} \Leftrightarrow \\
	\frac{|C_k| + 2\alpha * f(S_u)}{|A_k|+|B_k|+|C_k| + 2 * f(S_u)} & \geq \frac{|C_k|}{|A_k|+|B_k|+|C_k|} \\
	& = J_k(u,v)
\end{split}
\end{equation}

Therefore, our analysis shows that $J_k(u,v)$ is most likely increasing for random graphs with identical
degree distribution.

\subsection{Handling Updates}

In this section, we overview how our DBIndex is maintained when there are updates to the input graph.
There are two types of updates for graph data: updates to the attribute values associated with the nodes/edges and updates to the graph structure 
(e.g., addition/removal of nodes/edges).
Since the DBIndex is an index on the graph structure which is independent of the attribute values in the graph,
the DBIndex is not affected by updates to the graph's attribute values.

The efficient maintenance of the DBIndex in the presence of structural updates is challenging as a single structural change (e.g., adding an edge) could affect many vertex windows.  To balance the tradeoff between efficiency of index update and efficiency of query processing, 
we have adopted a two-phase approach to maintain the DBIndex.
The first phase is designed to optimize update efficiency where the DBIndex is updated incrementally whenever there are structural updates to the graph.
The incremental index update ensures that the updated index functions correctly but does not fully optimize the query efficiency of the updated index
in terms of maximizing the shared computations.
The second phase is designed to optimize query efficiency where the DBIndex is periodically re-organized to maximize share computations.

As an example of how the DBIndex is updated incrementally, consider a structural change where a new edge is added to the input graph.
Let $S$ denote the subset of graph vertices whose windows have expanded (with additional vertices) as a result of the insertion of the new edge.
Let $W'(v)$ denote the set of additional vertices in the vertex window of $v$ for each vertex $v \in S$.
Based on the identified changes to the vertex windows (i.e., $S$ and $\{W'(v) |\ v \in S\}$), 
we construct a secondary DBIndex which is then merged into the primary DBIndex.
As the identified changes are small relative to the entire graph and collection of vertex windows,
the construction and merging of the secondary index can be processed efficiently relative to an index reorganization to fully optimize query efficiency.

%% file: sec5_parent_index.tex
\begin{figure}[t]
\centering
\includegraphics[width=60mm,height=30mm]{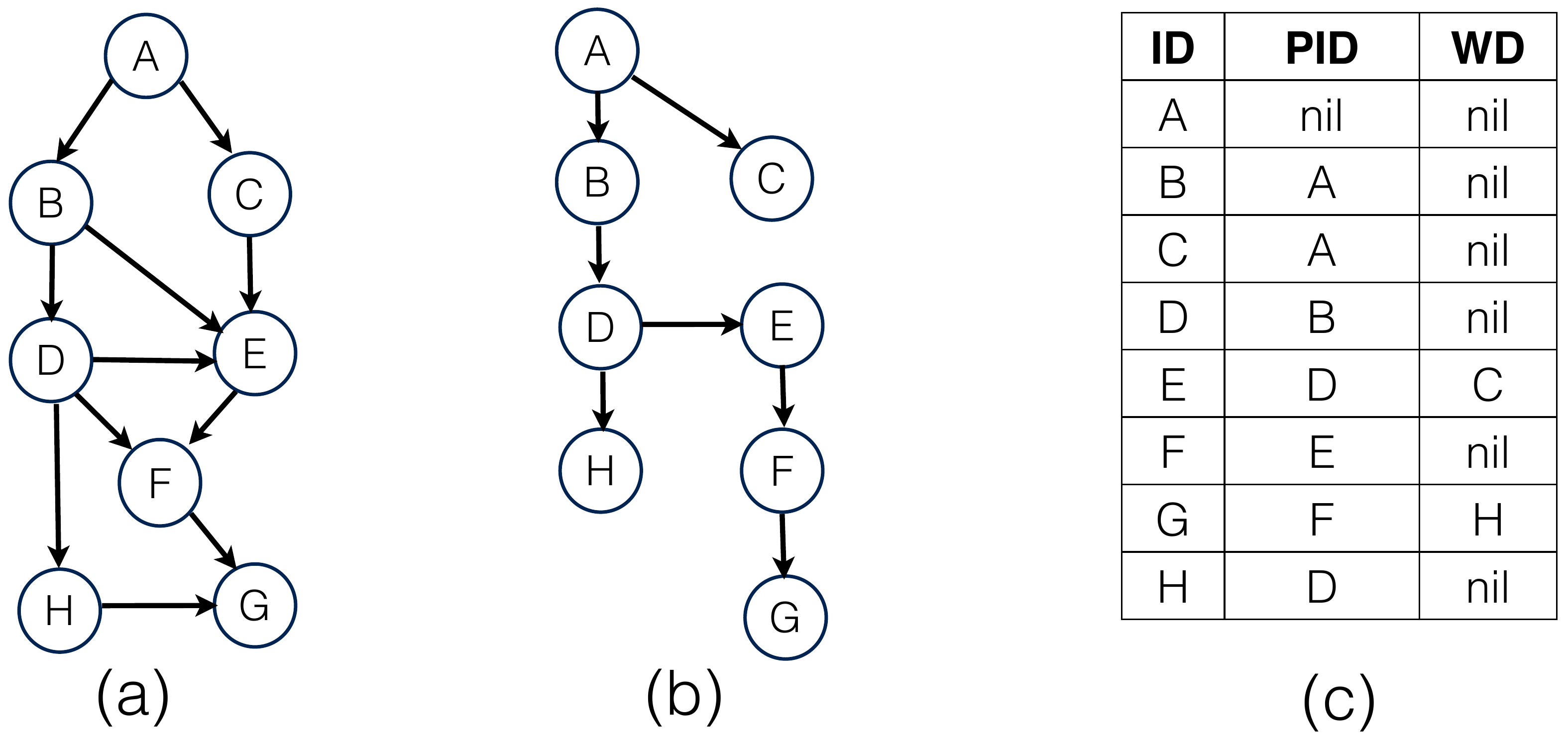}
	\caption{I-Index Construction over the Pathway DAG in Fig.~\ref{fig:topological}. (a) shows the DAG structure; (b) provides the inheritance relationship discovered during the index construction; (c) shows the final I-Index.}
	\label{fig:diff-index}
\end{figure}

\section{Inheritance Index}
\textit{DBIndex} is a general index that can support both k-hop as well as
topological window queries. 
However, the evaluation of a topological window function, $W_t$, 
can be further optimized due to its containment feature. 
In other words, the window of a descendant vertex 
completely covers that of one of its ancestors. 
This feature can be formally formulated in the following theorem. 

\begin{thm}
\label{thm:containment}
In a DAG, if vertex u is the ancestor of vertex v, the topological window of $v$, $W_t(v)$ completely contains the window of $u$, $W_t(u)$, i.e., $W_t(u) \subset W_t(v)$.   
\end{thm}

\begin{proof}
In a DAG, if $u$ is the ancestor of $v$, then $u \leadsto v$. $\forall w \in W_t(u)$, then $w \leadsto u$. As $u \leadsto v$, then $w \leadsto v$. Thus, $w \in W_t(v)$ and the theorem is proved.   
\end{proof}

Let us consider the BioPathway graph in Fig.~\ref{fig:topological} as
an example.
Fig.~\ref{fig:diff-index} (a) shows its abstract DAG. In (a), 
D is the ancestor of E. In addition, we can see that the window of $D$, 
$W_t(D)$ is $\{A, B, D\}$ and the window of $E$, $W_t(E)$ is $\{A, B, D, C, E\}$. It is easy to see that $W_t(D) \subset W_t(E)$. 

Now, Theorem~\ref{thm:containment} provides us with opportunities for optimizing
the space and computation of topological window queries. 
First, since the set of vertices corresponding to the
window of a node, say $u$, is a superset of the set of vertices of its
parent node, say $v$, there is no need to maintain the full set of 
vertices of the window at $u$. Instead, 
we only need to maintain the difference between 
$W_t(u)$ and $W_t(v)$. We note that in a DAG, it is possible for $u$ to have
multiple parents, $v_1, \cdots, v_k$. In this case, the parent
which has the smallest difference with $u$ can be used; where there
is a tie, it is arbitrarily broken.
We refer to this parent as the {\em closest} parent. For instance, in Fig.~\ref{fig:diff-index} (a), instead of maintaining 
$\{A, B, D, C\}$ for $W_t(E)$, it can simply maintain the difference 
to $W_t(D)$ which is $\{C\}$. This is clearly more space efficient.

Second, using a similar logic, the aggregate computation
at a node $u$ can actually reuse the 
aggregate result of its closest parent, $v$.
Referring to our example, the aggregation result of $W_t(D)$ can be 
simply passed or inherited to $W_t(E)$ and further aggregated with the difference 
set ($\{C\}$) in $W_t(E)$ to generate the aggregate value for $W_t(E)$. 
Fig.~\ref{fig:diff-index} (b) indicates the inheritance relationship that 
the values of the father can be inherited to the child in the tree. 

Thus, we propose a new structure, called the \textbf{inheritance index}, 
$I$-$Index$, to support efficient processing of topological window queries. 
In $I$-$Index$, each vertex $v$ maintains two information. 
\begin{itemize}
\item The first information is the ID of the closest parent (say $u$) 
of $v$. We denote this as PID($v$).
\item The second information is the difference between 
$W_t(v)$ and $W_t(u)$. We denote this as WD($v$). 
\end{itemize}

With PID($v$), we can retrieve $W_t(u)$, and combining with 
WD($v$), we can derive $W_t(v)$. 
Likewise, we can retrieve the aggregation result of $u$ 
which can be reused to compute $v$'s aggregation result.

Fig.~\ref{fig:diff-index}(c) shows the I-index of our example
in Fig.~\ref{fig:diff-index}(a). In the figure, I-Index
is represented in a table format; 
the second column is the PID and the third indicates the WD. 

\subsection{Index Construction} 

Building an I-Index for a DAG
can be done efficiently. 
This is because the containment relationship can be easily 
discovered using a topological scan.
Algorithm~\ref{algo:piconstruction} lists the pseudo code for 
index creation. 
The scheme iterates through all the vertices in a topological order.
For vertex $v$, the processing involves two steps.
In the first step, we determine the closest parent
of $v$. This is done by comparing the cardinalities of 
the windows of $v$'s parents, and 
find the parent with largest value. 
The corresponding PID is recorded in the \emph{PID} field of 
I-Index (Lines~\ref{algo:tp-step1-start}-\ref{algo:tp-step1-stop}). 
In the second step, the window of $v$, $W_t(v)$, is pushed to 
its children 
(Lines~\ref{algo:tp-step2-start}-\ref{algo:tp-step2-stop}). 
When the processing of $v$ finishes, its window can be discarded. This
frees up the memory space, which makes the scheme memory efficient.

We not that the complexity of Algorithm~\ref{algo:piconstruction} is non-trivial to analyze. 
This dues to the difficulty of analyzing of the number of ancestors of each vertex. Suppose the 
average number of ancestors for each vertex is $H$, then Algorithm~\ref{algo:piconstruction} is of
complexity $O(H|V|*d)$, where $d$ is the average degree of the graph. This complexity is close to the
output complexity. That is to gather the all vertex-window mapping, at least $O(H|V|)$ elements needs
to be outputted. Thus the indexing time complexity is reasonably efficient.

We further note that the size of \emph{I-Index} is hard to be precisely evaluated. 
This dues to the difficulty of analyzing the window difference. Assume the average size of window difference
is $D$, then the size of \emph{I-Index} is $O(D|V|)$. Although $D$ can be as large $O(|V|)$, our 
experimental results indicate that the index size is always 
comparable to the graph size. We defer this discussion to 
section \ref{sec:experiments}. Furthermore, it is possible to
reduce the index size (should it be a concern) by employing
compression techniques. 

\begin{algorithm}[h]
\label{alg:piconstruction}
\caption{CreateI-Index}
\begin{algorithmic}[1]
\REQUIRE Input graph: $G$ 
\ENSURE Inheritance Index: $IIndex$ 
\STATE $IIndex \leftarrow ()$
\STATE $p \leftarrow ()$ \COMMENT{stores the window for each vertex}
\STATE $c \leftarrow ()$ \COMMENT{stores the cardinality of window for each vertex}
\FORALL{$v \in$ topological order}
\STATE $WD \leftarrow -\infty$ \COMMENT{the window difference}
\STATE $bestu \leftarrow nil$
\FORALL{$u \in v.parent$} \label{algo:tp-step1-start}
	\IF{$c[u] > diff$} 
		\STATE $diff \leftarrow c[u]$
		\STATE $bestu \leftarrow u$
	\ENDIF
\ENDFOR \label{algo:tp-step1-stop}
\STATE $IIndex[v].WD \leftarrow WD$
\STATE $IIndex[v].PID \leftarrow bestu$
\STATE $p[v] \leftarrow p[v] \cup v$
\FORALL{$u \in v.child$} \label{algo:tp-step2-start}
\STATE $p[u] \leftarrow p[u] \cup p[v]$
\ENDFOR \label{algo:tp-step2-stop}
\STATE $c[v] \leftarrow |p[v]|$ \COMMENT{update window cardinality}
\STATE $p[v] \leftarrow ()$ \COMMENT{release memory}
\ENDFOR
\end{algorithmic}
\label{algo:piconstruction}
\end{algorithm}


\subsection{Query Processing using I-Index}

By employing the I-Index, window aggregation can be processed efficiently 
for each vertex according to 
the topological order. Algorithm \ref{alg:dstw-index} 
provides the pseudo code for the query processing. 
Each vertex $v$'s window aggregation value can be calculated 
by using the following formula: 
\begin{equation}
\Sigma (W_t(v)) = \Sigma (W_t(v.PID),\Sigma(v.WD))
\end{equation}
where $\Sigma$ is the aggregate function. As the vertex is 
processed according to the topological order, $W_t(v.PID)$ 
would have already been calculated while processing $v$'s parent and 
thus can be directly used for $v$ without any recompution. 
In general, $v$'s window aggregation is achieved by utilizing its 
parent's aggregate value and window difference sets. 
This avoids repeated aggregate computation and achieves 
the goal of computation sharing between a vertex and its parent. 
In so doing, the computation overhead can be further reduced. 
Take the index  provided in Fig.~\ref{fig:diff-index} (c) as an example, 
assume the query wants to calculate the sum value over each window 
for every vertex. As a comparison, the number of add operations 
are 33, 22, 16 for the cases without any index, 
with DBIndex and with I-Index index respectively.

\begin{algorithm}
\caption{QueryProcessingOverIIndex}
\begin{algorithmic}[1]
\REQUIRE Input graph $G$, aggregate function $\Sigma$, inheritance index $IIndex$ 
\ENSURE $w$ \COMMENT{The aggregation result of each vertex}
\STATE $w \leftarrow ()$
\FORALL{$v \in$ topological order} \label{code:dstw-index-tp1}
\STATE $u \leftarrow IIndex[v].PID$
\STATE $WD \leftarrow IIndex[w].WD$
\STATE $S \leftarrow v.val$
\STATE $S \leftarrow \Sigma(S, w[u])$
\FORALL{$t \in WD$ } \label{code:dstw-index-tp2}
\STATE $S\leftarrow \Sigma(S, t.val)$
\ENDFOR
\STATE $w[v] \leftarrow S$
\ENDFOR
\RETURN $w$
\end{algorithmic}
\label{alg:dstw-index}
\end{algorithm}

As the query processing in Algorithm~\ref{alg:dstw-index} basically scans the \emph{I-Index}, the query
complexity essentially correlates to the index size. As we shown in the experiment session, the query can be performed efficiently 
in various graph conditions. We defer the discussion to Section 6.

\subsection{Handling Updates}
We cover the updates handling in this section. For attribute updates, \emph{I-Index} is not affected since \emph{I-Index} is only structure related. Structure updates on \emph{I-Index} consists of node updates and edge updates. It is easy to handle the case where an isolated node is added or delete, since it does not affect any other nodes in the graph. Adding (resp. deleting) a node with edges can be done via edge insertions (resp. deletions). Here we focus on describing single edge insertion and deletion. We use $I$ to denote the \emph{I-index} and $I(v)$ to denote the index entry of $v$. 

During an update of edge $e(s,t)$, there are two types of vertices are affected. The first type contains single node $t$, which is the endpoint of edge $e$. The second type of nodes contains all the descendants of $t$. There are altogether four special cases needs to be consider during the updates. We illustrate the four cases with aids of Fig.~\ref{fig:dag_update}. In Fig.~\ref{fig:dag_update}, the dashed edge $e(s,t)$ is the edge to be updated (added or deleted). The node $u$ is the \emph{lowest common ancestor}(LCA) of $t$ and $s$. The cloud shape $A$ and $B$ are nodes in between of $u,s$ and $u,t$. Since $u$ is the \emph{LCA}, $A$ and $B$ are thus disjoint. Bold edge $e(c,d)$ indicate that $I(d).PID$ is $c$. We distinguish and handle the four cases as follows:

\textbf{Case I ($I(t).PID = s$)}: As shown in Fig.~\ref{fig:dag_update} (a), during deletion, $t$ needs to choose a parent from $A$ to be its $I(t).PID$. $I(t).WD$ needs to be updated accordingly. In this case, no insertion needs to be considered since if there is no edge between $s$ and $t$, $I(t).PID$ cannot be $s$. 

\textbf{Case II ($I(t).PID \neq s$)}: As shown in Fig.~\ref{fig:dag_update} (b), during insertion, $B$ needs to be excluded from $I(t).WD$. During deletion, any node in $B$ that cannot reach $t$ needs to be included in $I(t).WD$. Since $A$ and $B$ are disjoint, every node in $B$ needs to be removed from $I(t).WD$

\textbf{Case III ($t \leadsto I(v).PID$)}: As shown in Fig.~\ref{fig:dag_update} (c), during insertion, any node in $B$ needs to be removed from $I(v).WD$. During deletion, any node in $B$ that reaches $v$ but cannot reach $a$ needs to be added to $I(v).WD$.

\textbf{Case IV ($t \nrightarrow I(v).PID$)}: As shown in Fig.~\ref{fig:dag_update} (d), during insertion, any node in $B$ that cannot reach $r$ needs to be included into $I(v).WD$. During deletion, any node in $B$ that cannot reach $v$ needs to be excluded from $I(v).WD$.

During structure updates, essential operations are computing $A,B$ and performing reachability queries. Computing $A,B$ can be done via existing techniques such as \cite{bender2005lowest,czumaj2007faster} while reachability query can be supported by indexing methods such as \cite{yildirim2013dagger}. We defer exploring for more efficient updating algorithms to future work

\begin{figure*}[t]
\centering
\begin{subfigure}{0.20\linewidth}
\centering
  \includegraphics[width=\textwidth]{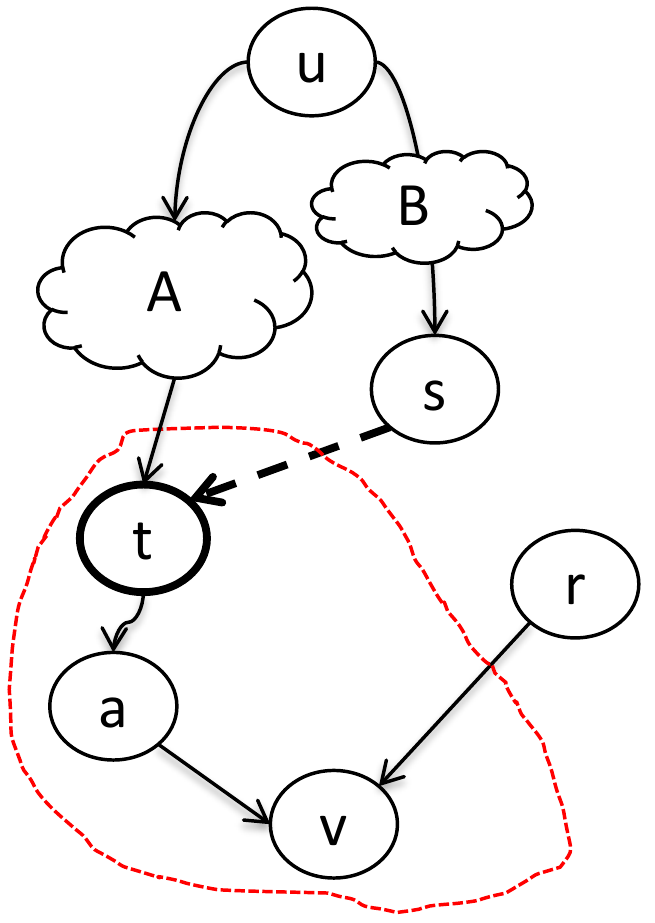}
  \caption{Edge Update Case 1}
\end{subfigure}%
\begin{subfigure}{0.20\linewidth}
\centering
  \includegraphics[width=\textwidth]{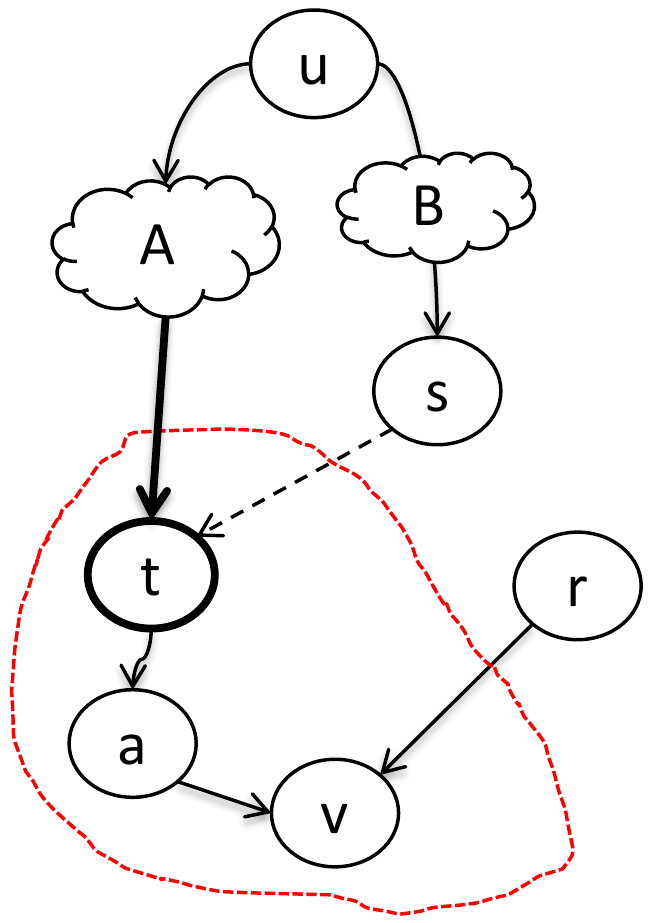}
  \caption{Edge Update Case 2}
\end{subfigure}
\begin{subfigure}{0.19\linewidth}
\centering
  \includegraphics[width=\textwidth]{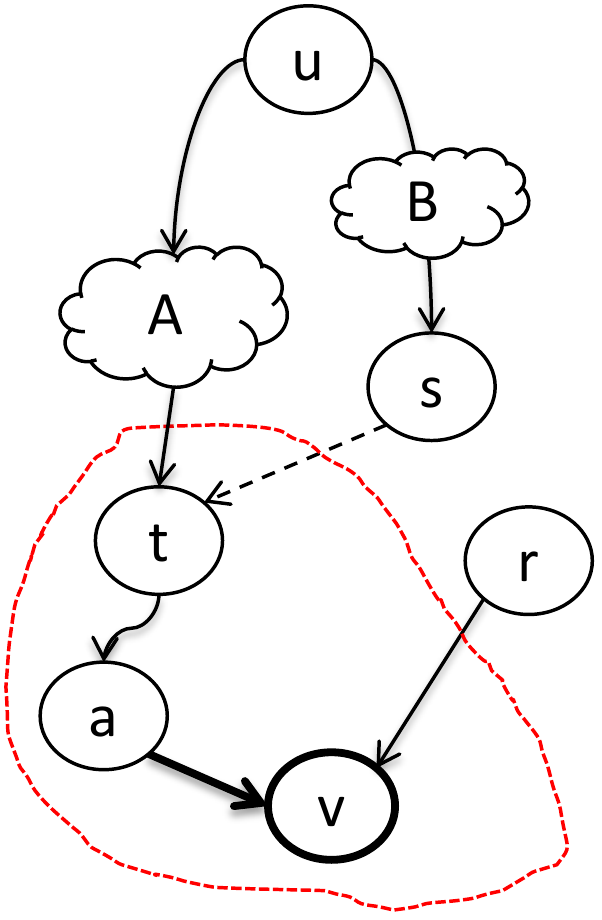}
  \caption{Edge Update Case 3}
\end{subfigure}
\begin{subfigure}{0.20\linewidth}
\centering
  \includegraphics[width=\textwidth]{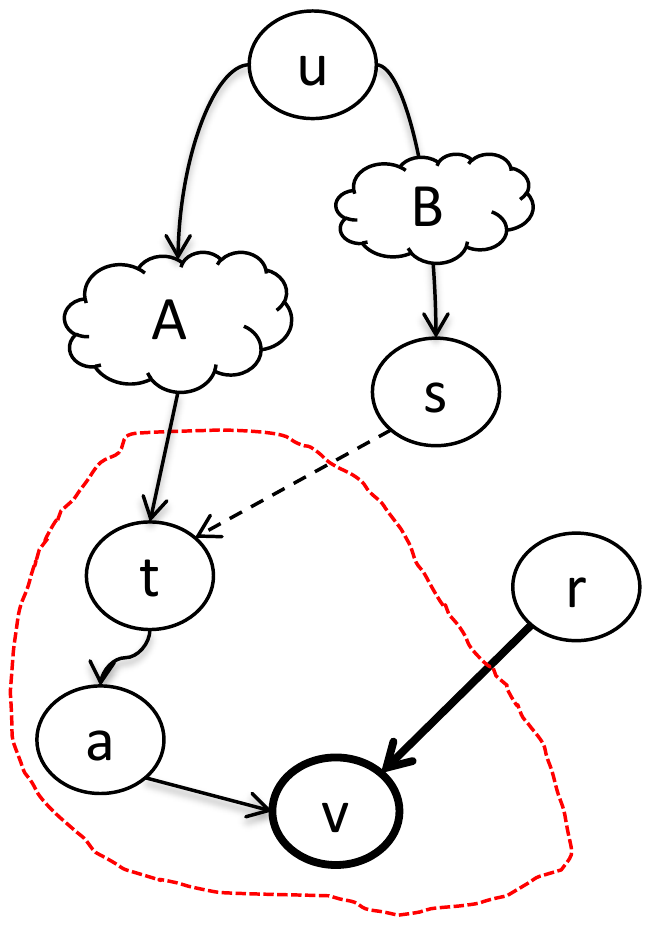}
  \caption{Edge Update Case 4}
\end{subfigure}
\caption{Updates on \emph{I-Index}. Cloud shape indicate the nodes in the subgraph between the endpoint nodes. The dashed circle indicate the affected range of updates. The bold arrow indicates the $PID$ field \emph{I-Index}}
\label{fig:dag_update}
\end{figure*}

%% file: sec6_experiment.tex
\section{Experimental Evaluation}\label{sec:experiments}
In this section, we present a comprehensive experimental evaluation 
of our solutions using several real-world information networks and 
various synthetic datasets. Since the focus of this paper is on query processing efficiency, 
we do not evaluate the efficiency of index updates. 
All experiments are conducted on an 
Amazon EC2 r3.2xlarge machine\footnote{http://aws.amazon.com/ec2/pricing/}, 
with an 8-core 2.5GHz CPU, 60GB memory and 320GB hard drive 
running with 64-bit Ubuntu 12.04. As the source code of EAGR is not
available, we implemented it and used it as a reference in our comparative
study.  
All algorithms are implemented in Java and run under JRE 1.6.

\begin{table}[h]
\begin{tabular}{|l|l|l|l|}
\hline 
\rule[-1ex]{0pt}{2.5ex} Name & Type & \# of Vertices & \# of Edges \\ 
\hline 
\rule[-1ex]{0pt}{2.5ex} LiveJournal1 & undirected & 3,997,962 & 34,681,189 \\ 
\hline 
\rule[-1ex]{0pt}{2.5ex} Pokec & directed & 1,632,803 & 30,622,564 \\ 
\hline 
\rule[-1ex]{0pt}{2.5ex} Orkut & undirected & 3,072,441 & 117,185,083 \\ 
\hline 
\rule[-1ex]{0pt}{2.5ex} DBLP & undirected & 317,080 & 1,049,866 \\ 
\hline 
\rule[-1ex]{0pt}{2.5ex} YouTube & undirected & 1,134,890 & 2,987,624 \\ 
\hline 
\rule[-1ex]{0pt}{2.5ex} Google & directed & 875,713 & 5,105,039 \\ 
\hline 
\rule[-1ex]{0pt}{2.5ex} Amazon & undirected & 334,863 & 925,872 \\ 
\hline 
\rule[-1ex]{0pt}{2.5ex} Stanford-web & directed & 281,903 &  2,312,497 \\ 
\hline 
\end{tabular}
\caption{Large Scale Real Data}
\label{tab:realdata}
\end{table}

\textbf{Datasets.} For real datasets, we use 8 information networks 
which are available at the Stanford \emph{SNAP} 
website \footnote{http://snap.stanford.edu/snap/index.html}: 
LiveJournal1, Pokec, Orkut, DBLP, YouTube, Google, Amazon and Stanford-web. 
The detail description of these datasets is provided in 
Table~\ref{tab:realdata}. 

For synthetic datasets, we use two widely used graph data generators. 
We use the \emph{DAGGER} generator \cite{yildirim2013dagger} to generate 
all the synthetic DAGs and the SNAP graph data generator at the
Stanford SNAP website to generate non-DAG datasets. For each dataset, each vertex is associated with an integer attribute.

\textbf{Query.} In all the experiments, the window query is conducted 
by using the SUM() as the aggregate function over the integer attribute in each dataset. 

\begin{figure*}[t]
\centering
\begin{subfigure}{0.23\textwidth}
  \includegraphics[width=\textwidth]{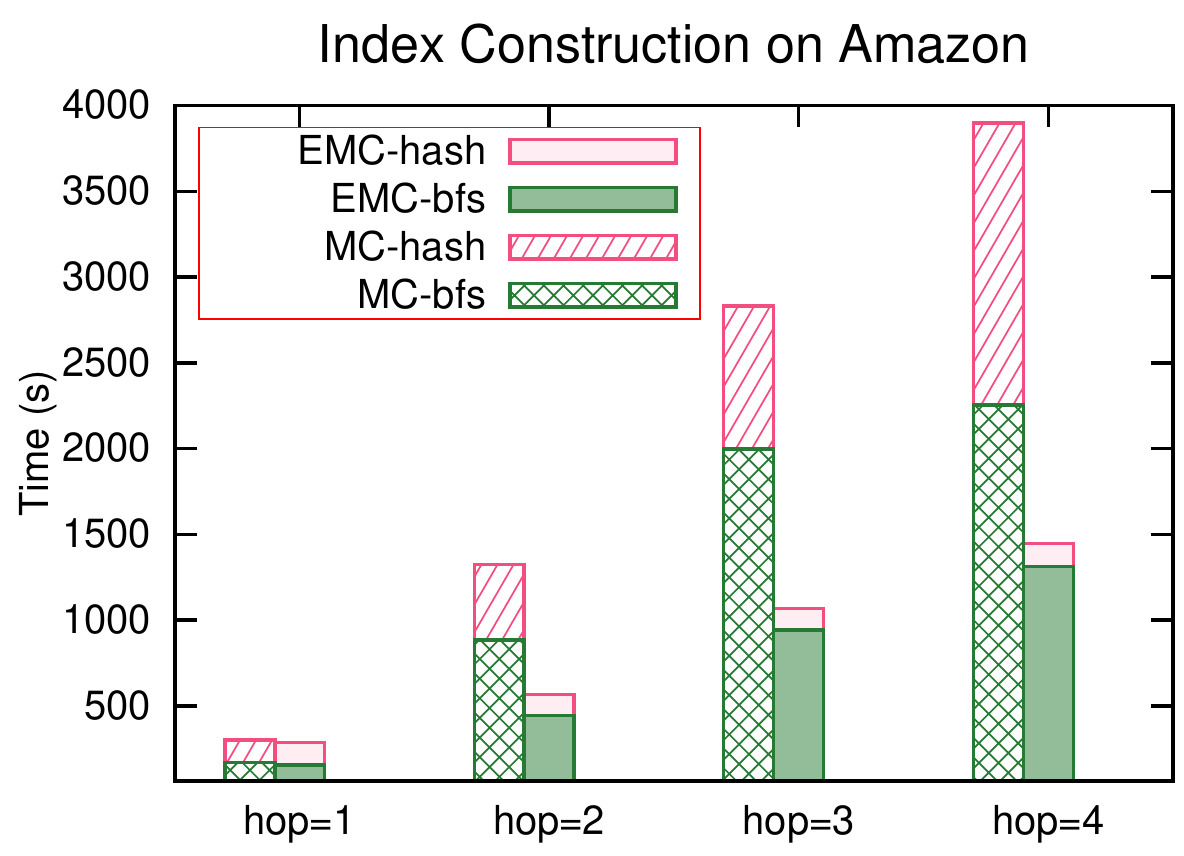}
  \caption{Index Built on Amazon}
\end{subfigure}
\begin{subfigure}{0.23\textwidth}
  \includegraphics[width=\textwidth]{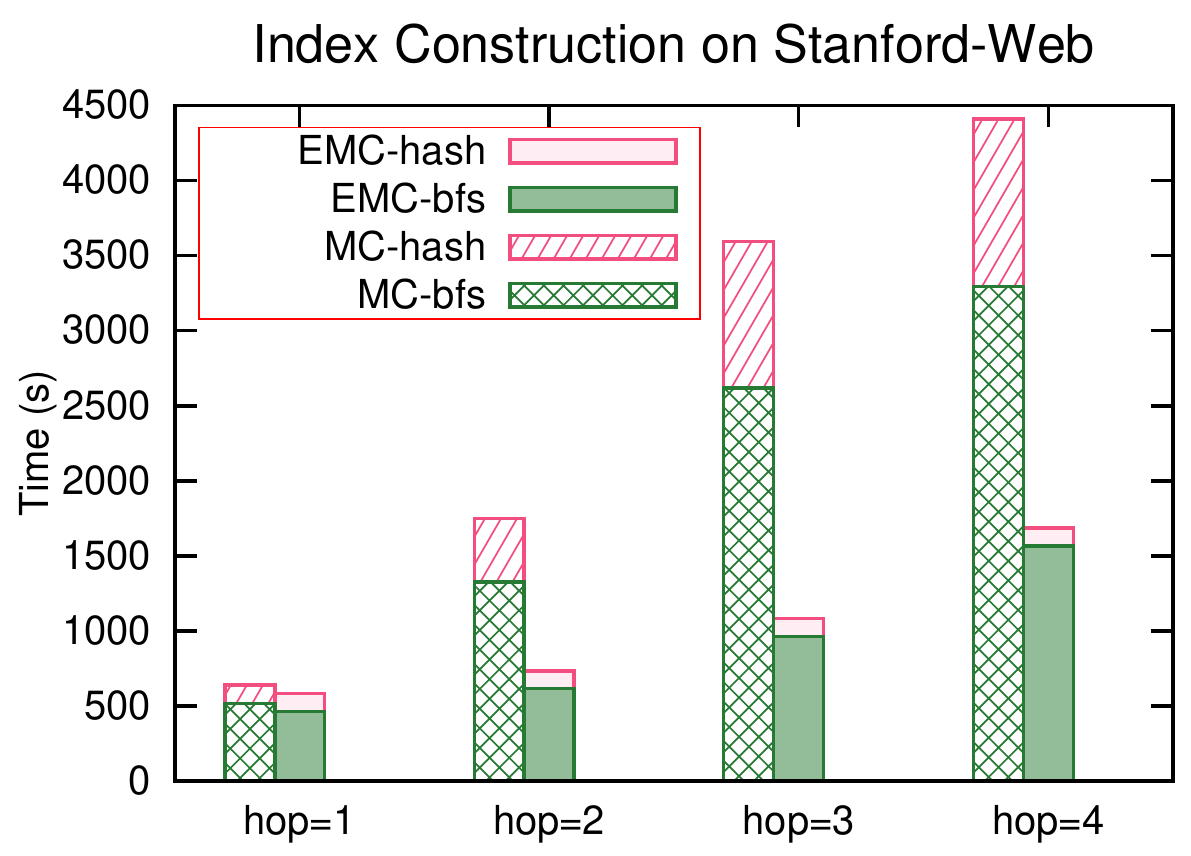}
  \caption{Index Built on Stanford-web}
\end{subfigure}
\begin{subfigure}{0.23\textwidth}
  \includegraphics[width=\textwidth]{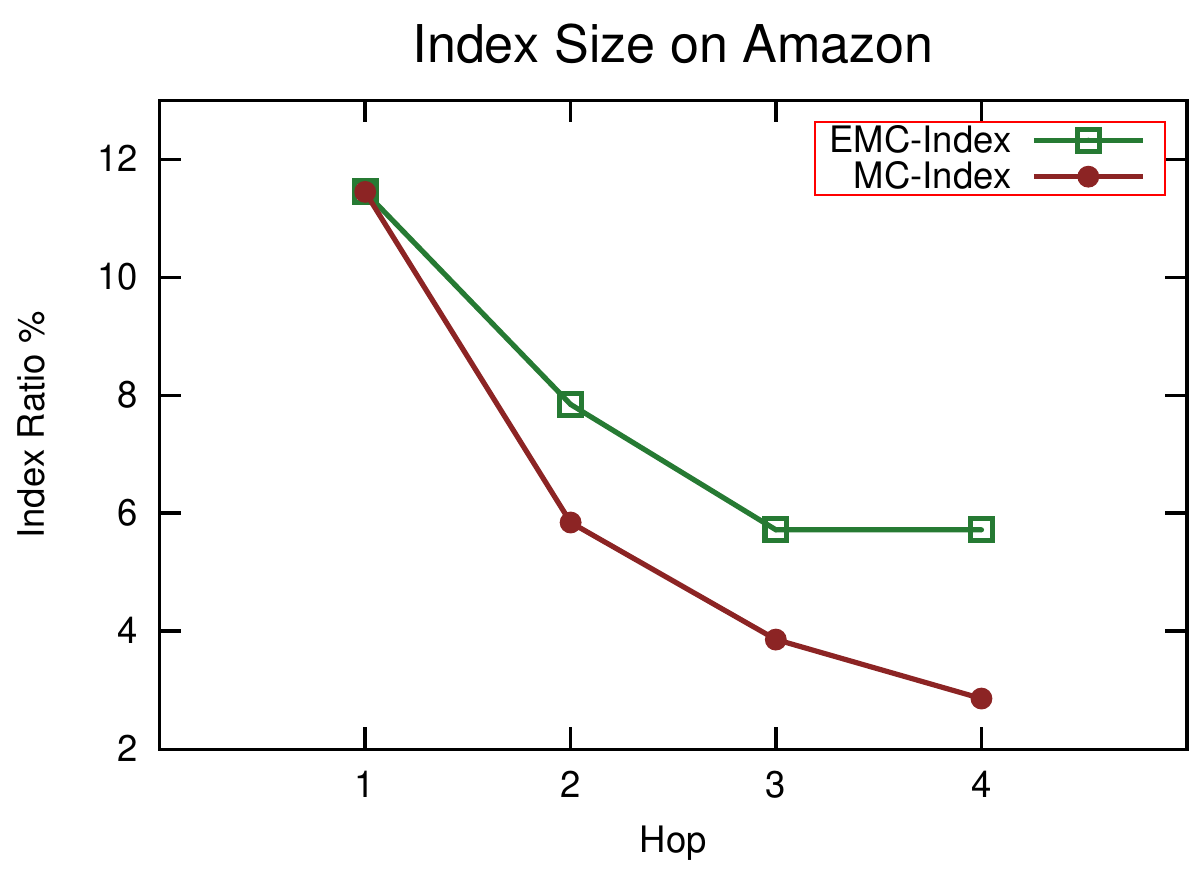}
  \caption{Index Size on Amazon}
\end{subfigure}
\begin{subfigure}{0.23\textwidth}
  \includegraphics[width=\textwidth]{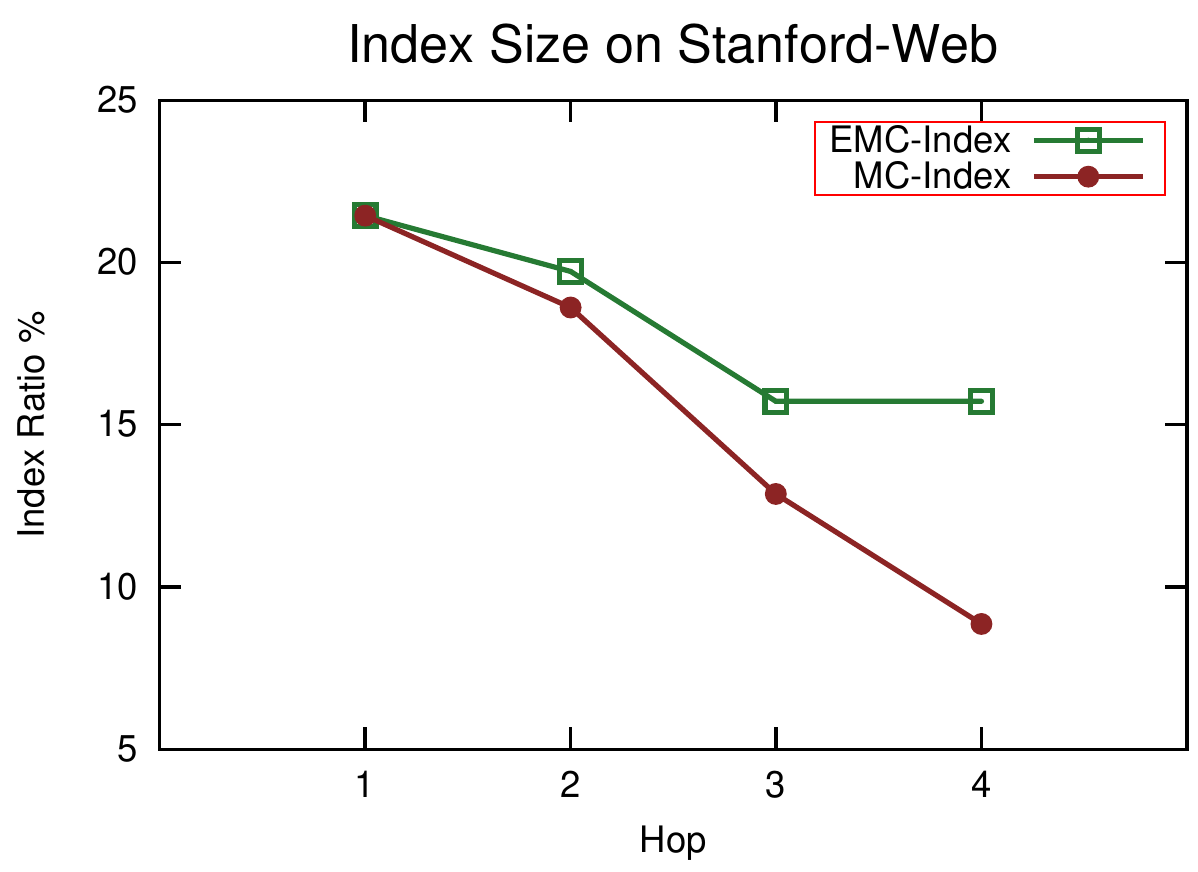}
  \caption{Index Size on Stanford-web}
\end{subfigure}
\caption{Index Construction Analysis for EMC and MC. (a) and (b) 
depict the index time for the Amazon and Stanford-web networks; 
(c) and (d) shows the index size for the Amazon and Stanford-web datasets}
\label{fig:index_analysis_emc_mc}
\end{figure*}

\subsection{Comparison between MC and EMC}
We first compare the effectiveness of the two DBIndex 
construction algorithms: MinHash Clustering (MC) and Estimated 
MinHash Clustering (EMC). We look at 
the index construction time, index sizes and query performance. 
All these experiments are conducted based on two real world datasets: 
Amazon and Stanford-web.For both datasets, we run a series of k-hop queries.\footnote{For the Stanford-web graph, which is directed,
the k-hop windows are directed k-hop windows where $u \in W(k)$ if there is a directed path
of at most $k$ hops from vertex $v$ to vertex $u$.
} For queries with hop count larger than 1, 
EMC uses 1-hop information for the initial clustering. 

\textbf{Index Construction. } Figs.~\ref{fig:index_analysis_emc_mc} 
(a) and (b) compare the index construction time between MC and EMC 
when we vary the windows from 1-hop to 4-hop for the Amazon and Stanford-web 
graphs respectively. To better understand the time difference, 
the construction time is split into two parts: 
the MinHash cost (EMC-hash or MC-hash) and the breadth-first-search traversal (to compute the k-hop window)  cost (EMC-bfs or MC-bfs). The results show the same trend 
for the two datasets. We made several observations.
First, as the number of hops increases, the indexing time increases as well. 
This is expected as a larger hop count results in a larger window size 
and the BFS and computation time increase correspondingly. 
Second, as the hop count increases, the difference between the
index time of EMC and that of MC widens. 
For instance, as shown in Figs.~\ref{fig:index_analysis_emc_mc}(a) and (b), for the 4-hop window queries, compared to MC, 
EMC can save $62\%$ and $66\%$ construction time for the Amazon and Stanford-Web datasets respectively.
EMC benefits from both the low MinHash cost and low BFS cost. 
From Figs.~\ref{fig:index_analysis_emc_mc} (a) and (b), 
we can see that the MinHash cost of MC increases as the number of hops 
increases, while that for EMC remains almost the same as the 1-hop case. 
This shows that the cost of MinHash becomes more significant for larger windows. 
Thus, using 1-hop clustering for larger hop counts reduces the MinHash cost 
in EMC. Similarly, as EMC saves on BFS cost for k-hop queries where $k > 1$, 
the BFS cost of EMC is much smaller than that of MC as well. 

\textbf{Index Size.} Figs.~\ref{fig:index_analysis_emc_mc} (c) and (d) 
present the effect of hop counts on the index size for the Amazon and Stanford-web 
datasets respectively. The y-axis shows the index ratio which is the index size over the original graph size. The insights we derive are: 
First, the index size is rather small compared to the original graph - it
varies from $3\%$ to $12\%$ of the original graph for the Amazon dataset 
and from $8\%$ to $22\%$ for the Stanford-web dataset. 
Second, the index size decreases as the number of hops increases. 
While this appears counter-intuitive initially, it is actually reasonable - a larger hop results in a bigger window, which leads
to more dense blocks. Third, the index ratio of EMC is slightly larger 
than that of MC for larger hop count. This indicates that MC can find more dense blocks 
than EMC to reduce the index size. Fourth, the index ratio on 
the Amazon dataset is much smaller than the ones on the Stanford-web dataset. 
This is because the Amazon dataset is undirected while the Stanford-web dataset
is directed. For the Stanford-web dataset, since we use directed k-hop windows, the window size is naturally smaller. 

\textbf{Query Performance.}       
Figs.~\ref{fig:khop_effective_query_time} (a) and (b)
present the query time of MC and EMC 
on the two datasets respectively 
as we vary the number of hops from 1 to 4.
To appreciate the benefits of an index-based scheme, we also implemented
a \emph{Non-indexed} algorithm which computes window aggregate by performing k-bounded breadth
first search for each vertex individually in real time.
In Figures~\ref{fig:khop_effective_query_time} (a) and (b),
the execution time shown on the y-axis is in log scale. The results show that the index-based 
schemes outperform the non-index approach by four orders of magnitude. 
For instance, for the 4-hop query over the Amazon graph, 
our algorithm is 13,000 times faster than the non-index approach. 
This confirms that it is necessary to have well-designed index support 
for efficient window query processing. By utilizing DBIndex, 
for these graphs with millions of edges, every aggregation query 
can be processed in just between 30ms to 100ms for the Amazon graph and 
between 60ms to 360ms for the Stanford-web graph. In addition, we can 
see that as the number of hops increases, the query time decreases. 
This is the case because a larger hop count eventually results in a larger
number of dense blocks where more (shared) computation can be salvaged. 
Furthermore, we can see that the query time of EMC is slightly longer 
than that of MC when the number of hops is large. This is expected as 
EMC does not cluster based on the complete window information; instead, it
uses only partial information derived from the 1-hop windows. 
However, the performance difference is quite small even for 4-hop queries- for
the Amazon dataset, the difference is only 20ms; and for the Stanford-web
graph, the difference is 35ms. 
For small number of hops, the time difference is even smaller. 
This performance penalty is acceptable as tens of milliseconds time 
difference will not affect user's experience.  As EMC is significantly more 
efficient than MC in index construction, EMC may still be a 
promising solution to many applications. As such, 
in the following sections, we adopt EMC for DBIndex in
our experimental evaluations.  
 
\begin{figure}[h]
\centering
\begin{subfigure}{0.48\linewidth}
\centering
  \includegraphics[width=\textwidth]{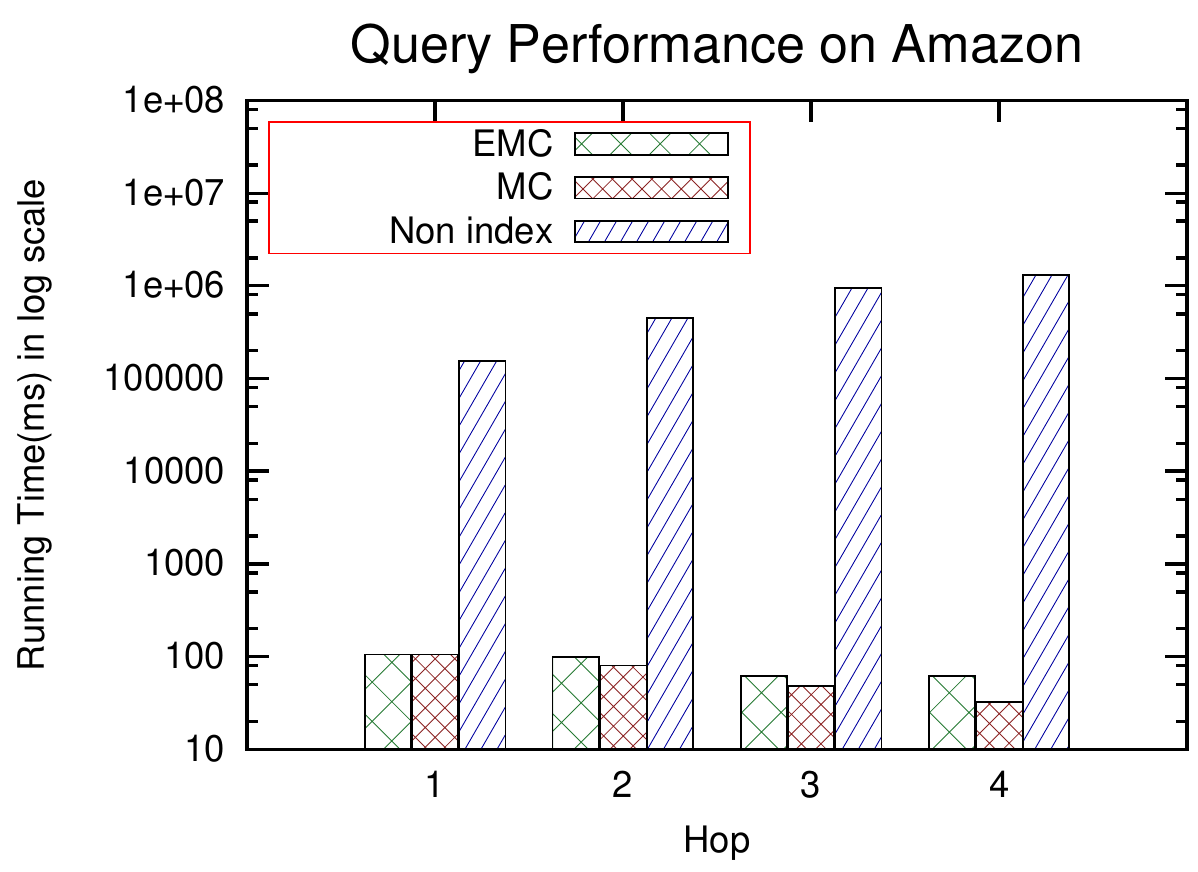}
  \caption{Amazon Graph}
\end{subfigure}%
\begin{subfigure}{0.48\linewidth}
\centering
  \includegraphics[width=\textwidth]{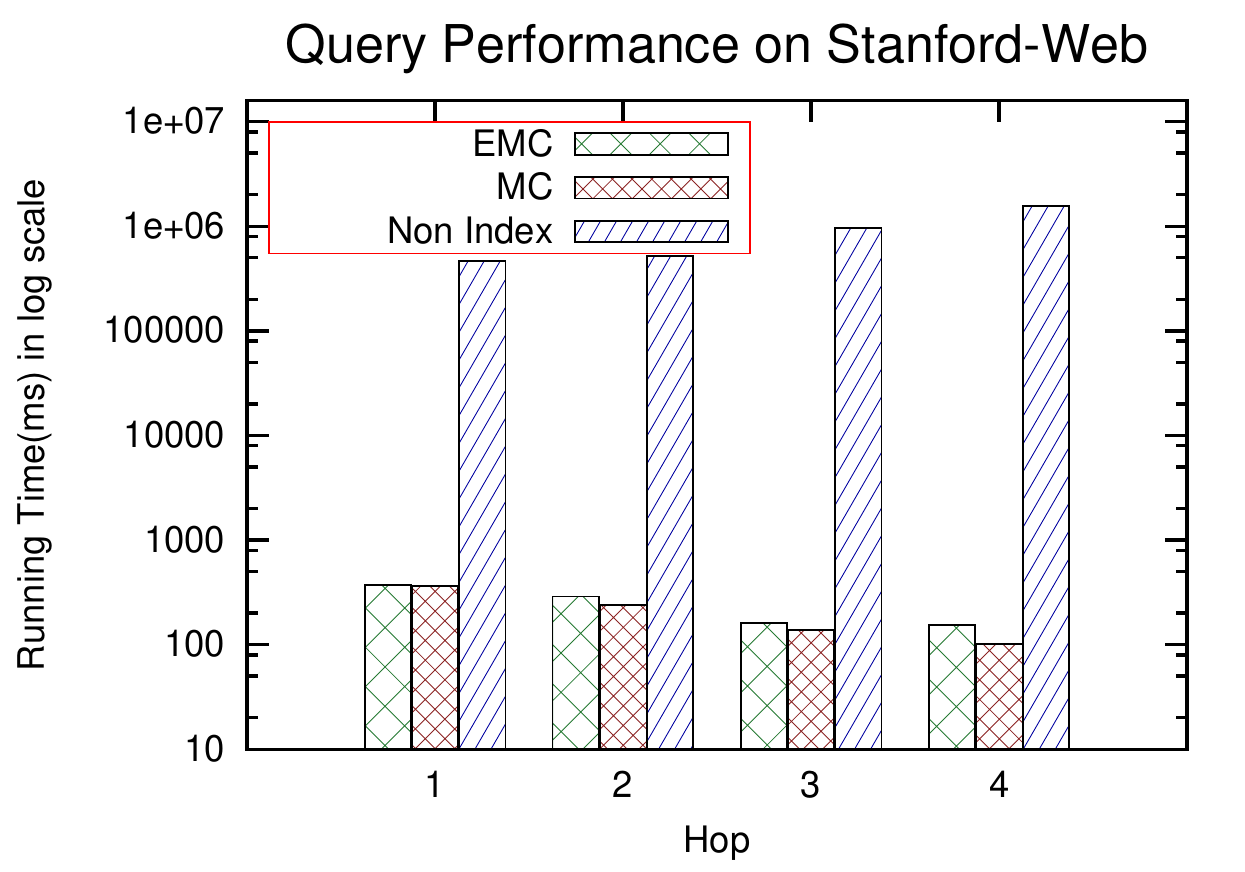}
  \caption{Stanford-web Graph}
\end{subfigure}
\caption{Query Performance Comparison of MC and EMC}
\label{fig:khop_effective_query_time}
\end{figure}

\subsection{Comparison between DBIndex and EAGR}

In this set of experiments, we compare DBIndex and 
EAGR \cite{mondal2014eagr} using both large-scale real 
and synthetic datasets. Like \cite{mondal2014eagr},
for each dataset, EAGR is run for 10 
iterations in the index construction.

\subsubsection{Real Datasets} 
We first study the index construction and 
query time performance of DBIndex and EAGR for 1-hop and 2-hop windows 
using 6 real datasets: DBLP, Youtube, Livejournal, Google, Pokec and Orkut. 
The results for 1-hop window and 2-hop window are presented
in Figs.~\ref{fig:1-hop-real} and~\ref{fig:2-hop-real} 
respectively. 
As shown in Figs.~\ref{fig:1-hop-real}(a) and~\ref{fig:2-hop-real}(a), both DBIndex and EAGR can build the index for all
the real datasets for 1-hop but EAGR ran out of the memory for 2-hop window queries on LiveJournal and Orkut datasets. This further confirms that EAGR incurs 
high memory usage as it needs to build the FPT and 
maintain the vertex-window mapping information. We also observe that 
DBIndex is significantly faster than EAGR in index creation. 
We emphasize that the time is shown in logarithmic scale. 
For instance, for Orkut dataset, EAGR takes 4 hours to build the index 
while DBIndex only takes 33 minutes. 

Fig.~\ref{fig:1-hop-real} (b) and Fig.~\ref{fig:2-hop-real} (b) 
show the query performance for 1-hop and 2-hop queries respectively. 
The results indicate that the query performance is comparable. 
For most of the datasets, DBIndex is faster than EAGR.
In some datasets (e.g. Orkut and Pokec), DBIndex performs 30\% faster than the EAGR. 
We see that, for 1-hop queries on 
Youtube and LiveJournal datasets and 2-hop queries on Youtube dataset, DBIndex is slightly slower than EAGR. We observe 
that these datasets are very sparse graphs where the intersections among windows are naturally small. For very sparse graphs, 
both DBIndex and EAGR are unable to find much computation sharing. In this case, the performance of DBIndex and EAGR is very close. For instance, in the worse case for livejournal, DBIndex is 9\% slower than EAGR where the actual time difference remains tens of millionseconds.    
 
Another insight we gain is that as expected, compared to Fig.~\ref{fig:1-hop-real} (b), 
2-hop query runs faster for both algorithms. This is because there is more computation sharing for 2-hop window query. 

In summary, 
DBIndex takes much shorter time to build but offers comparable, if not much faster, query 
performance than EAGR.

\begin{figure}[h]
\centering
\begin{subfigure}{0.48\linewidth}
  \centering
  \includegraphics[width=\textwidth]{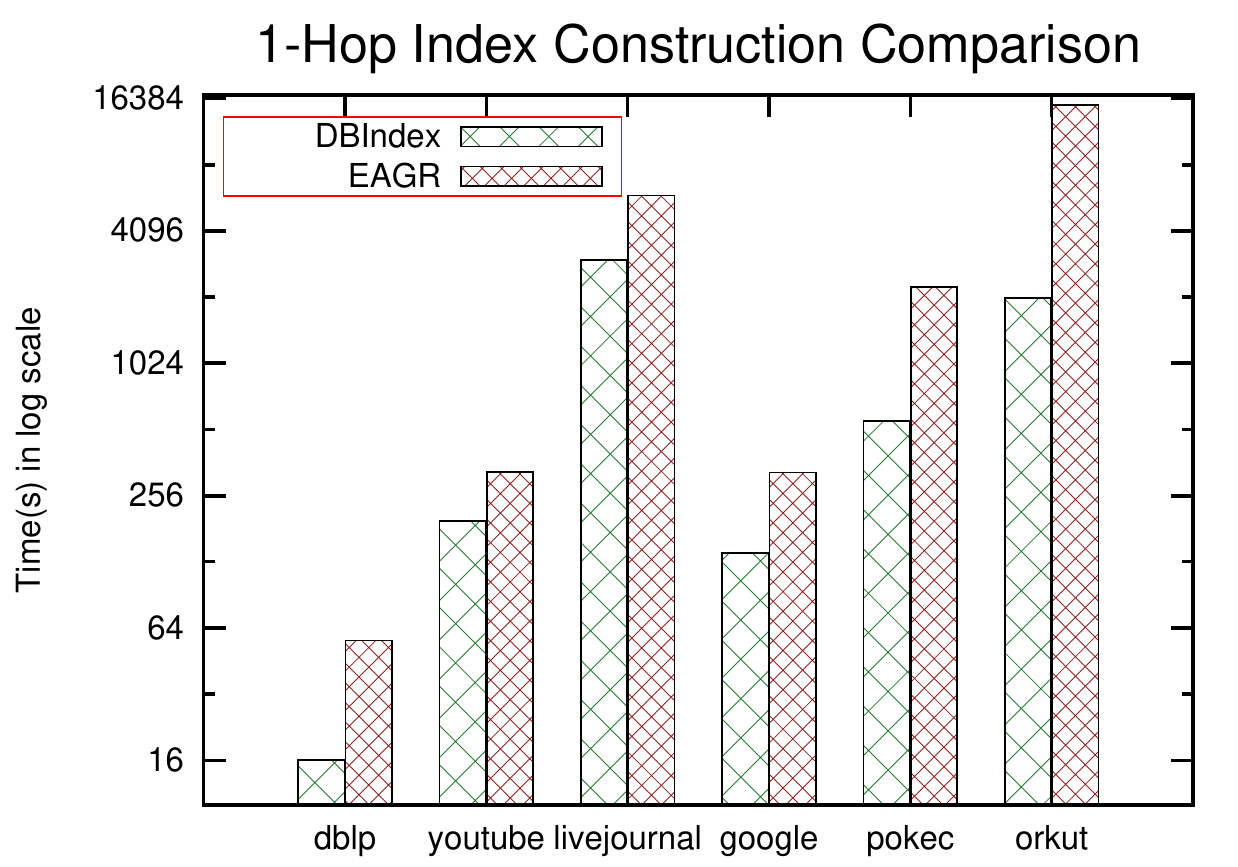}
  \caption{Index Construction}
\end{subfigure} \begin{subfigure}{0.48\linewidth}
  \centering
  \includegraphics[width=\textwidth]{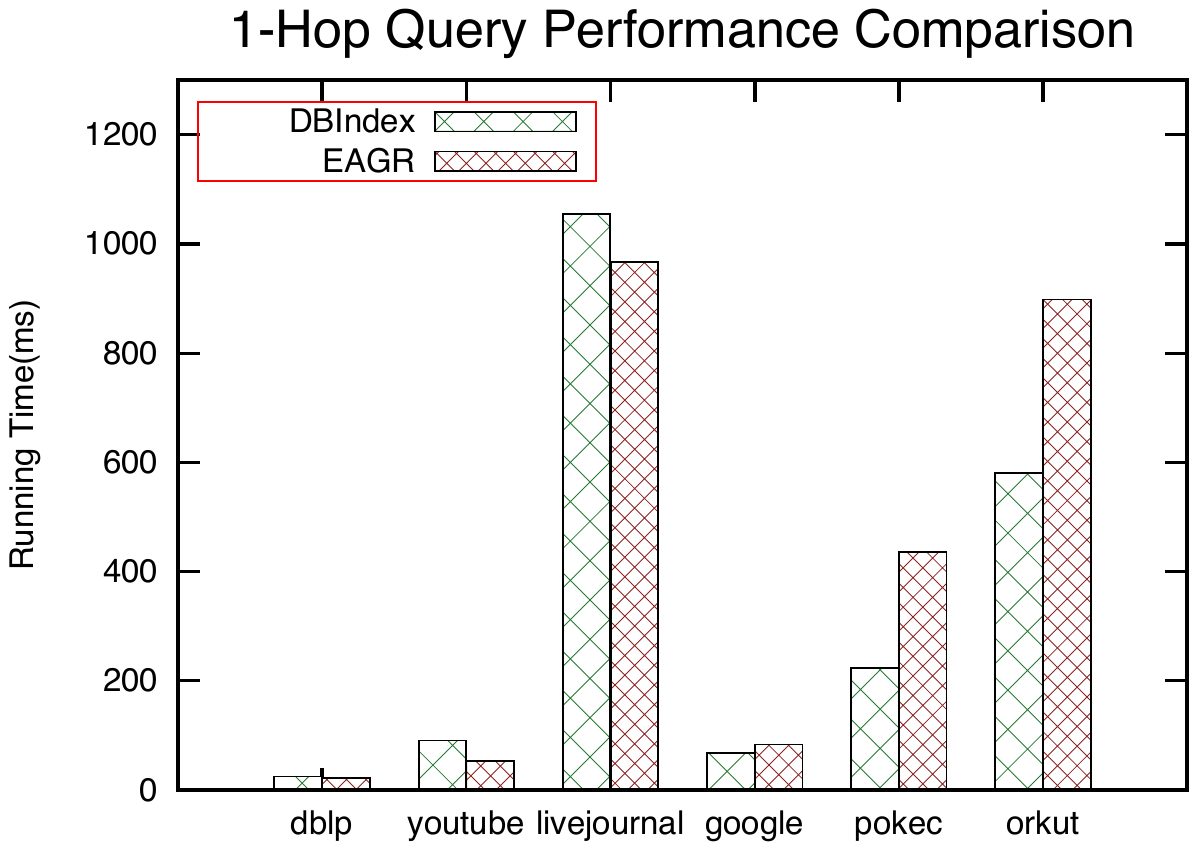}
  \caption{Query Performance}
\end{subfigure}%
\caption{Comparison between DBIndex and EAGR  for 1-hop query}
\label{fig:1-hop-real}
\end{figure}

\begin{figure}[h]
\centering
\begin{subfigure}{0.48\linewidth}
  \centering
  \includegraphics[width=\textwidth]{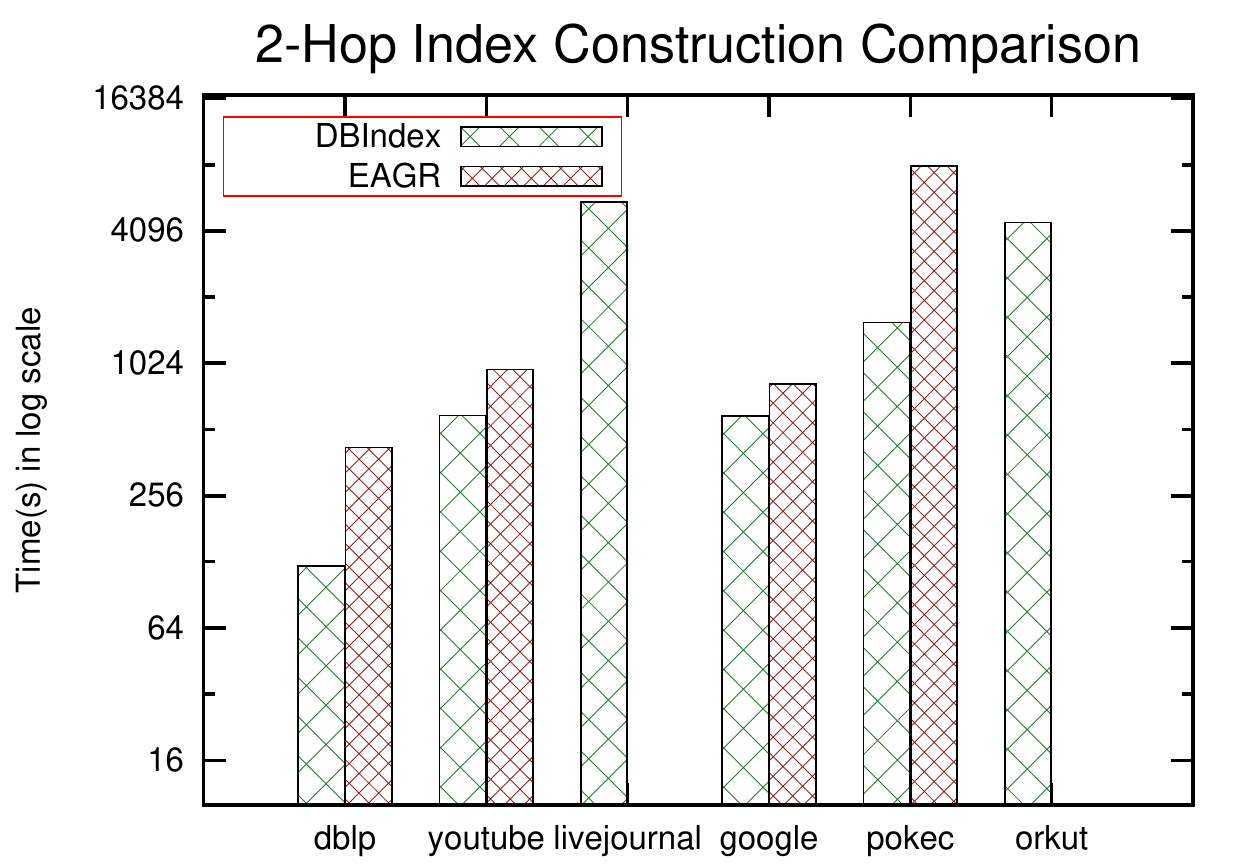}
  \caption{Index Construction}
\end{subfigure}
\begin{subfigure}{0.48\linewidth}
  \centering
  \includegraphics[width=\textwidth]{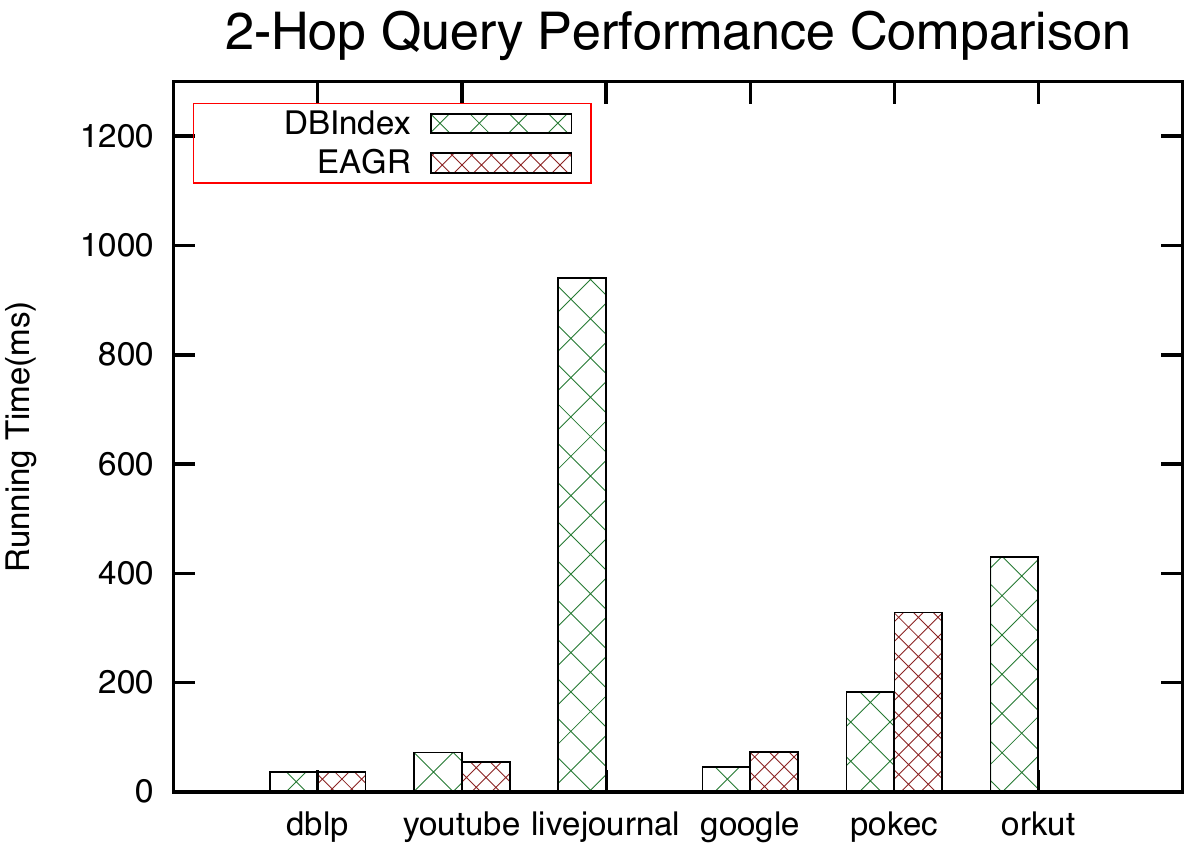}
  \caption{Query Performance}
\end{subfigure}
\caption{Comparison between DBIndex and EAGR for 2-hop query}
\label{fig:2-hop-real}
\end{figure}

\subsubsection{Synthetic Datasets}
To study the scalability of DBIndex under large-scale networks, 
we generated synthetic datasets using the SNAP generator. 


\textbf{Impact of Number of Vertices.} First, we study how the performance 
changes when we fix the degree \footnote{Degree means average degree of the graph. The generated graph is of Erdos-Renyi model } at 10 and vary the number of vertices from 
2M to 10M. Figs.~\ref{fig:khop_d10_h1} (a) and (b) show
the execution time for index construction and query performance respectively.
From the results, we can see that DBIndex outperforms EAGR in
both index construction and query performance. For the graph with 
10M vertices and 100M edges, the DBIndex query time is less than 
450 milliseconds. 
Moroever, when the number of vertices changes from 2M to 10M, 
the query performance only increases 3 times. This shows that
DBIndex is not only scalable, but offers acceptable performance..   

\begin{figure}[t]
\centering
\begin{subfigure}{0.48\linewidth}
  \centering
  \includegraphics[width=\textwidth]{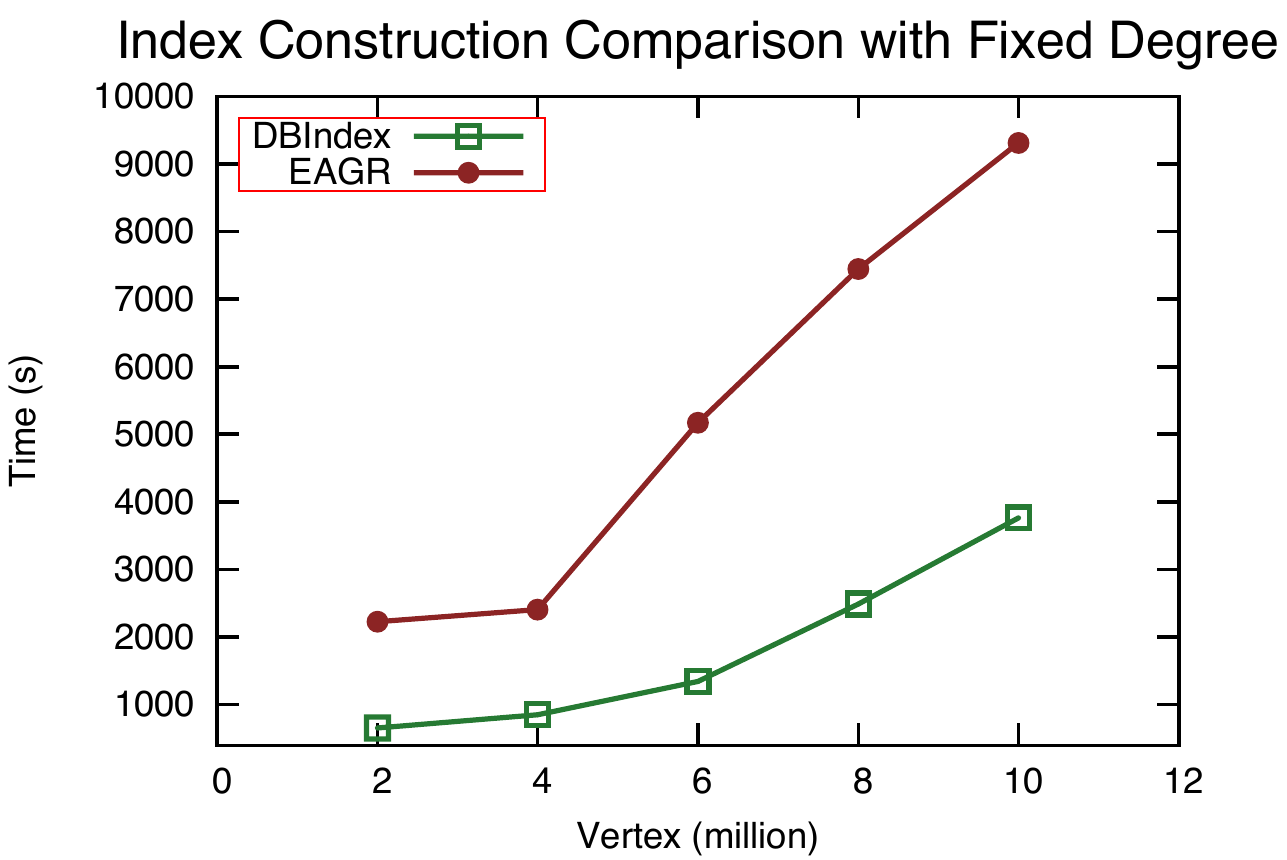}
  \caption{Index Construction}
\end{subfigure}
\begin{subfigure}{0.48\linewidth}
  \centering
  \includegraphics[width=\textwidth]{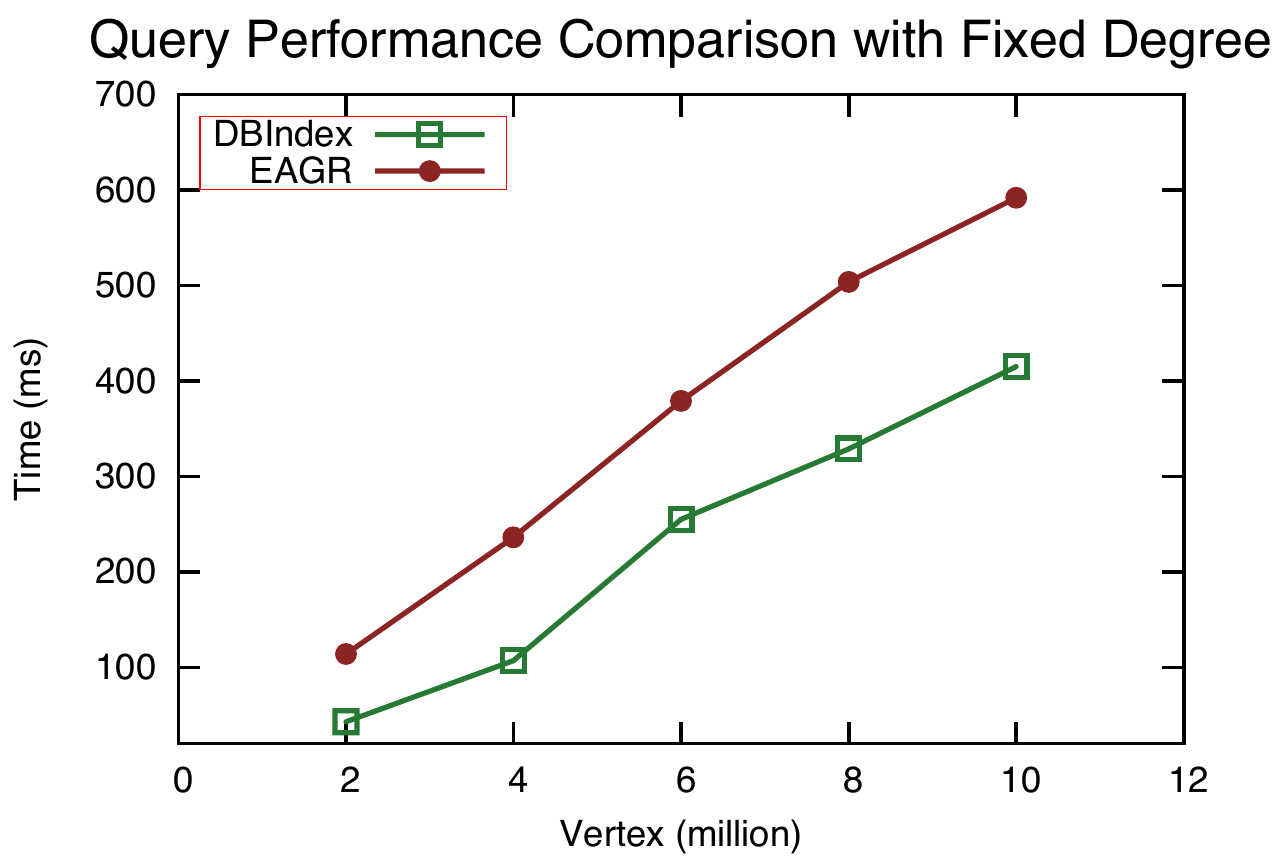}
  \caption{Query Performance}
\end{subfigure}%
\caption{Impact of number of vertices }
\label{fig:khop_d10_h1}
\end{figure}

\begin{figure}[t]
\centering
\begin{subfigure}{0.46\linewidth}
  \centering
  \includegraphics[width=\textwidth]{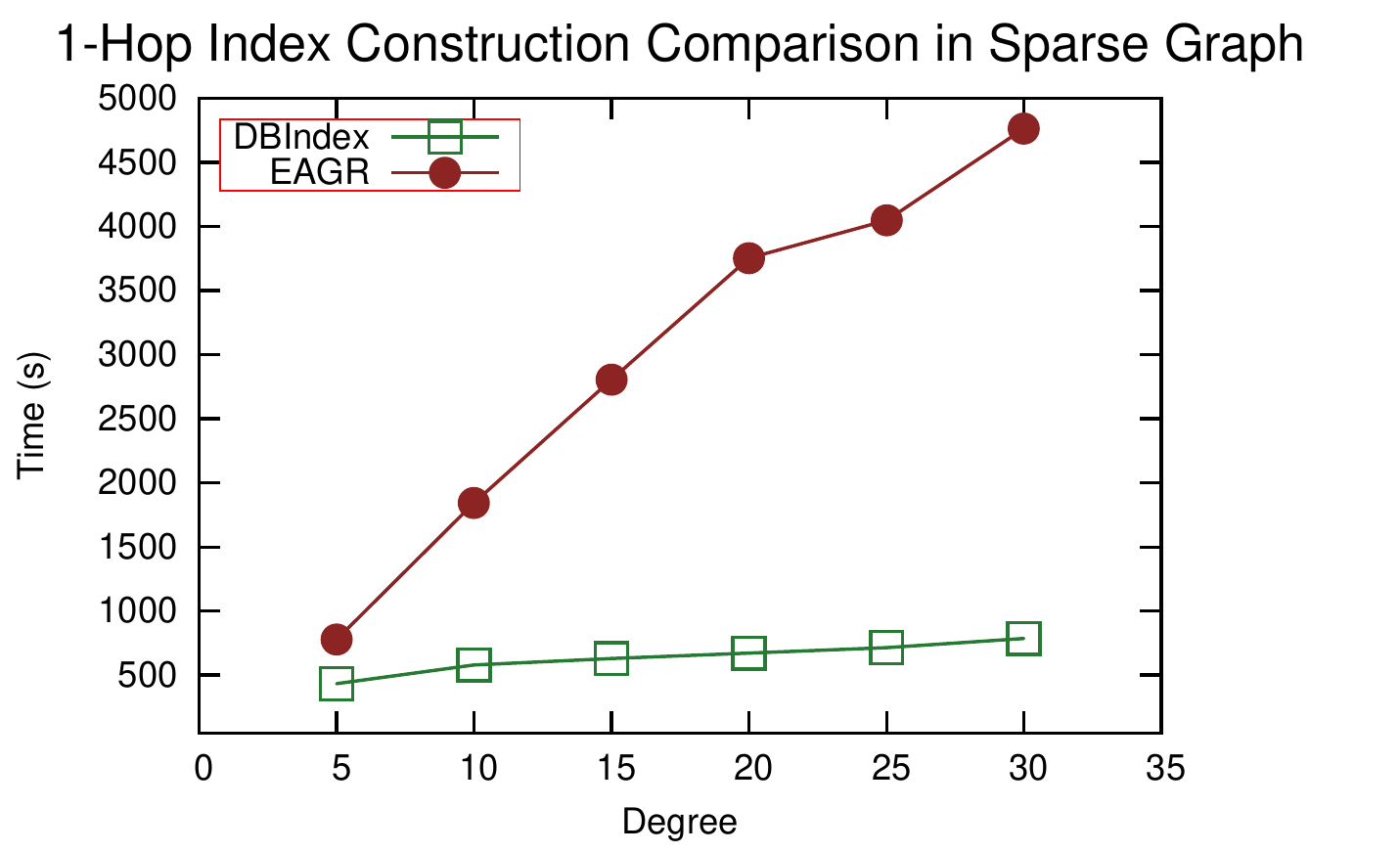}
  \caption{Index Construction}
\end{subfigure}
\begin{subfigure}{0.50\linewidth}
  \centering
  \includegraphics[width=\textwidth]{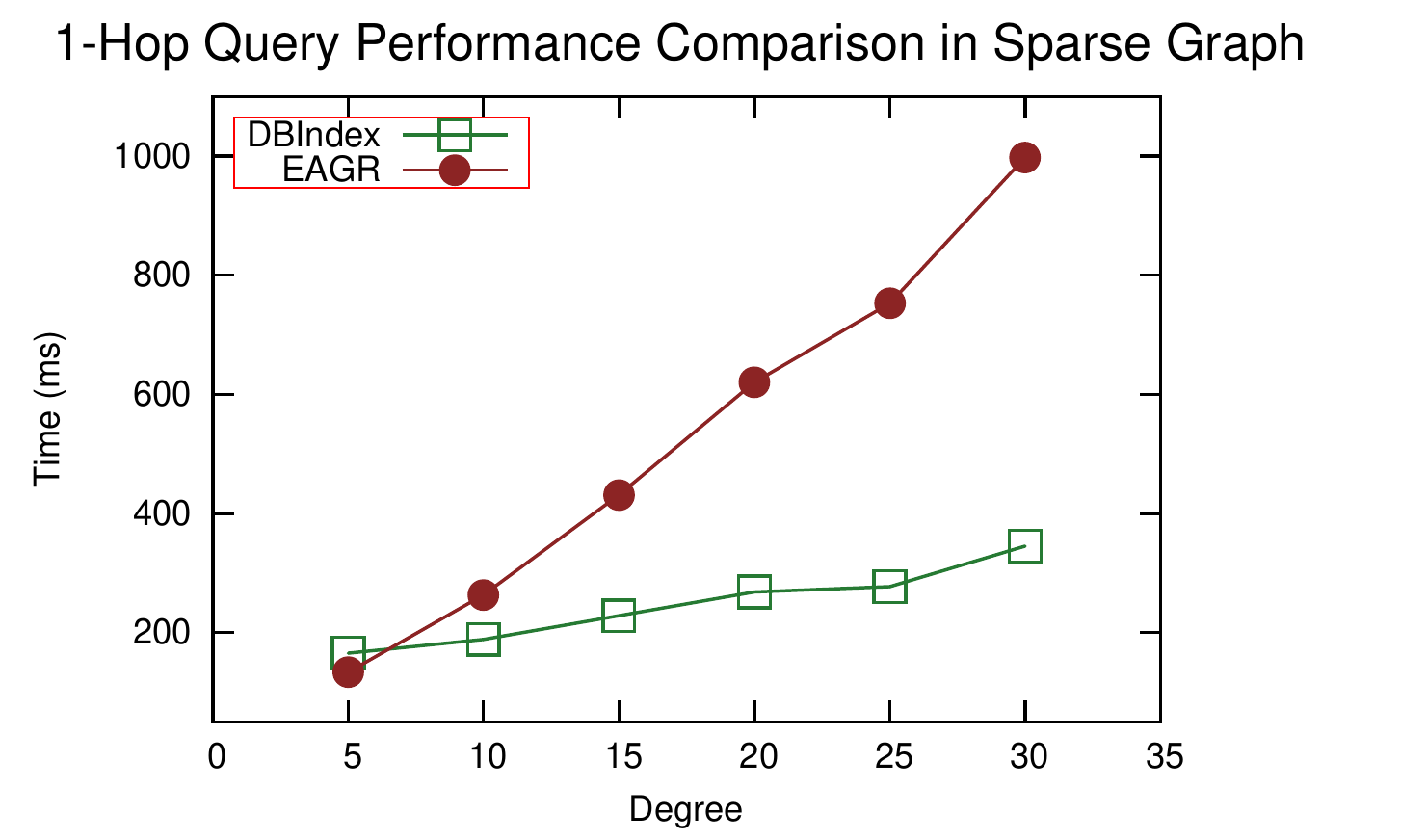}
  \caption{Query Performance}
\end{subfigure}
\begin{subfigure}{0.50\linewidth}
  \centering
  \includegraphics[width=\textwidth]{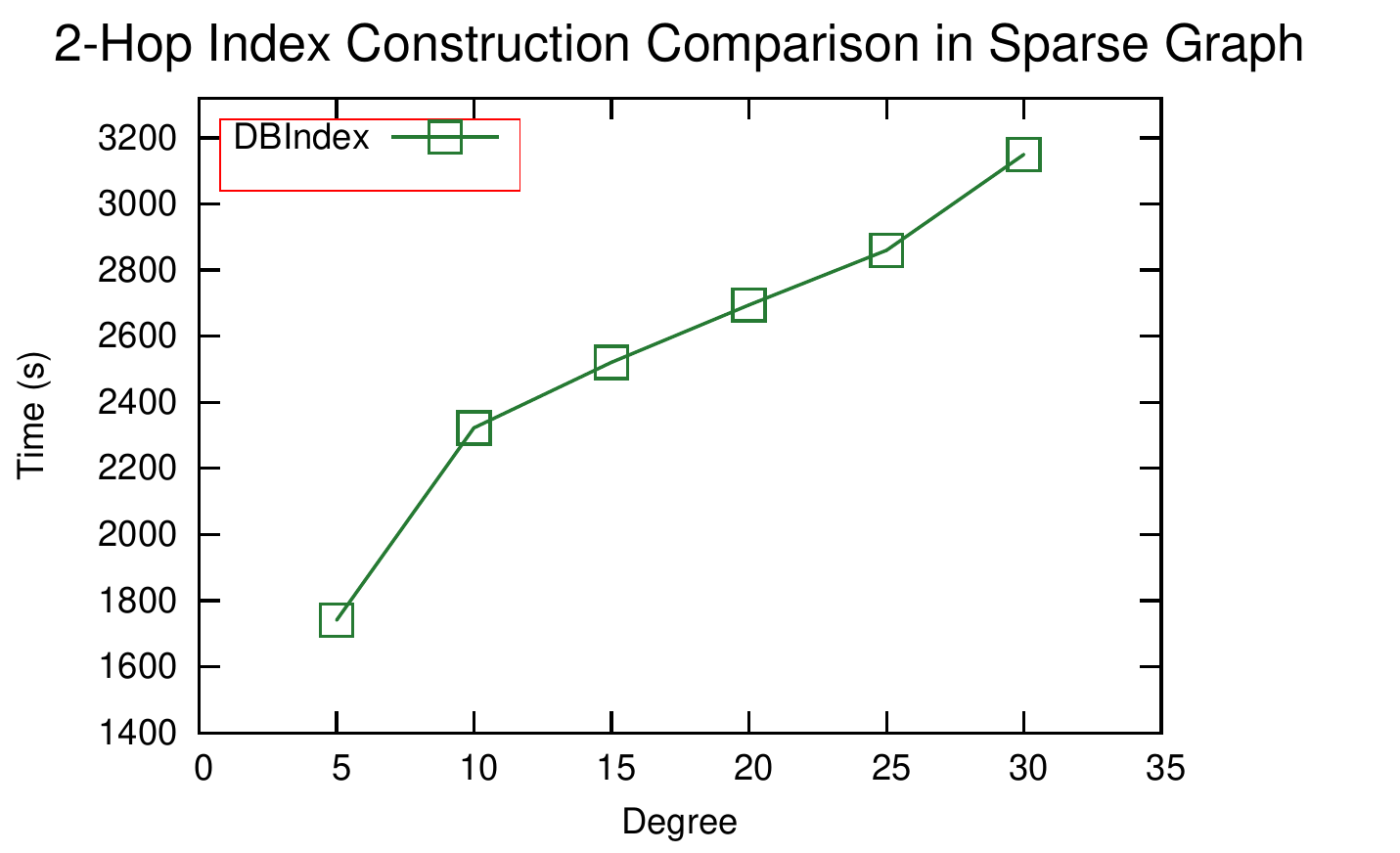}
  \caption{Index Construction}
\end{subfigure}
\begin{subfigure}{0.46\linewidth}
  \centering
  \includegraphics[width=\textwidth]{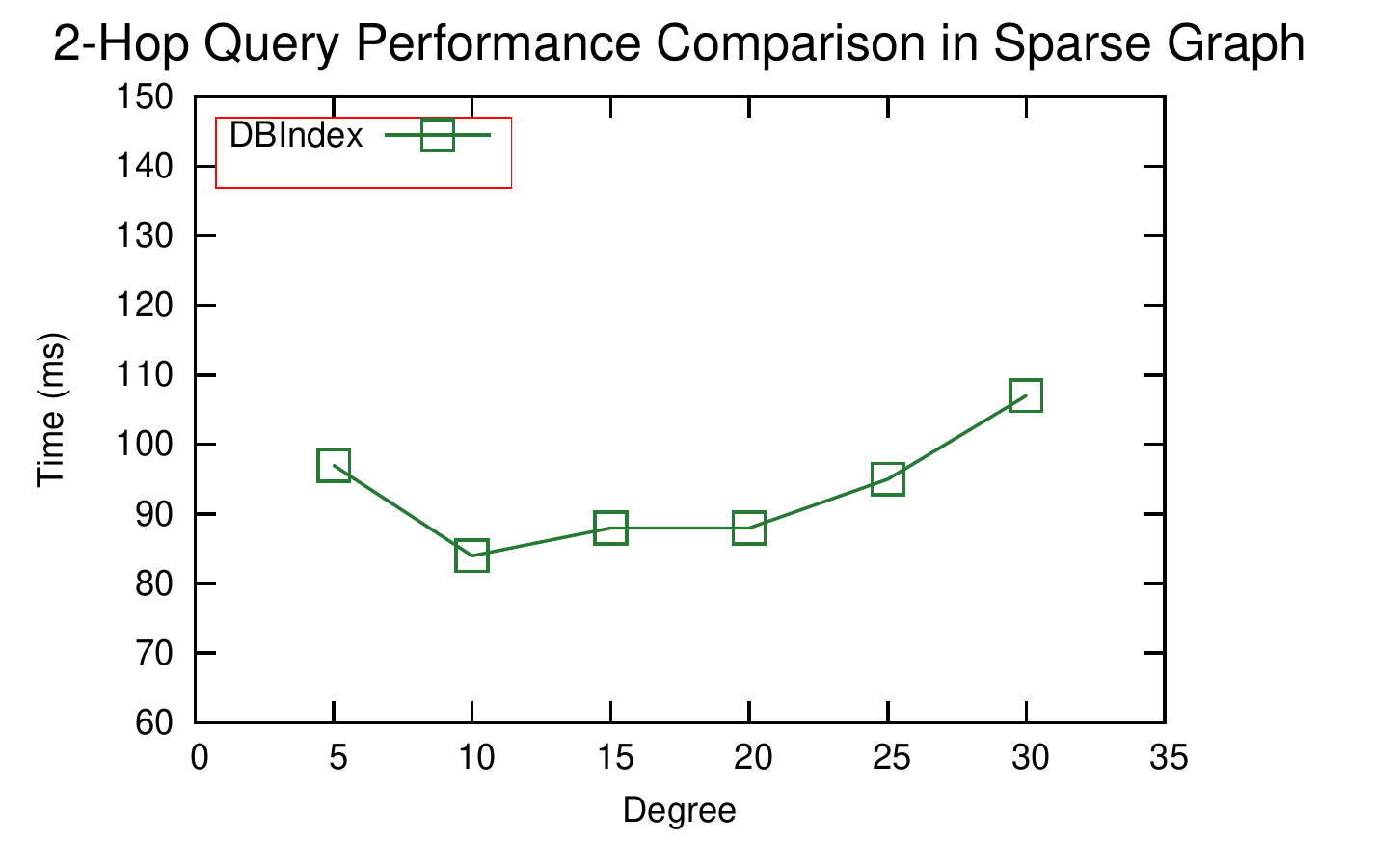}
  \caption{Query Performance}
\end{subfigure}
\caption{Impact of Degree over Sparse Graphs over 2M vertices. (a) and (b) are the results for 1-hop query; (c) and (d) are the results for 2-hop query. }
\label{fig:khop_v2m}
\end{figure}

\textbf{Impact of Degree over Sparse Graphs.} Our proposed DBIndex is effective when there is significant overlap between windows of neigboring nodes. As such, it is interesting to study how it performs for sparse graph where the nodes may not share many common neigbors. So, in these experiments, 
we study the impact of degree when the graph is relatively sparse.
We fix the number vertices of 2M and vary the vertex degree from 5 to 30. 
Figs.~\ref{fig:khop_v2m} (a) and (c) present the results 
on index construction for 1-hop and 2-hop queries respectively. 
For 1-hop queries, as degree increases, the time for index
construction also increases. 
However, the index creation time of DBIndex increases much 
slower than EAGR. This is because EAGR incurs relatively more 
overhead to handle multiple FPT creation and reconstruction. 
For 2-hop queries, EAGR failed to run. This is because even for 
a degree 5 sparse graph, the initial vertex-mapping can be as large 
as 90GB in a linked list manner, which exceeds the available memory. 
Note that the size becomes even larger when it is stored in a matrix manner.
Therefore, we can only show the results of DBIndex. 
In Fig.~\ref{fig:khop_v2m} (c), indexing time of DBIndex increases as the 
degree increases. This is expected as a bigger degree  
increases the overhead of graph traversal time to collect the window.

Fig.~\ref{fig:khop_v2m} (b) and (d)
show the results on query time 
for 1-hop and 2-hop queries respectively. 
We observe a similar pattern for the index construction time:
for 1-hop queries, the query time increases with increasing degree
but at a much slower rate than EAGR; for 2-hop queries, we observe in Fig.~\ref{fig:khop_v2m} (d) that the
query performance of DBIndex hovers around 100ms, which is much 
smaller than that of 1-hop query performance. 
This is because there are more dense blocks in the 2-hop case, 
in which case the query time can be faster compared to the 1-hop case. 

\textbf{Impact of Degree over Dense Graphs.} We study the 
impact of degree over very dense graphs with 200k vertices 
when the degree changes from 80 to 200. 
Figs.~\ref{fig:khop_v200k} (a) and (c) show
the execution time for index construction 
for 1-hop and 2-hop queries respectively. From the results, 
we can see that DBIndex performs well for dense graphs as well. 
As the degree increases, EAGR's performance degrades much faster than DBIndex. For 2-hop queries,
as shown in Figs.~\ref{fig:khop_v200k} (b) and (d), EAGR is only able to work on 
the dataset with degree 80 due to the memory issue. 
Even though the number of vertices is relatively small (only 200k), 
the number of edges is very large when the degree becomes big 
(e.g. 40M edges with degree of 200). 

Figs.~\ref{fig:khop_v200k} 
(b) and (d) show the results on query performance 
for 1-hop and 2-hop queries respectively. 
The results are consistent with that for sparse graphs
- DBIndex is superior over EAGR. 

In summary, the insight we obtain is that the scalability of EAGR 
is highly limited by its approach to build the index over 
the vertex-window mapping information. EAGR is limited 
by two factors: the graph size and the number of hops. 
DBIndex achieves better scalability as it does not need to 
create a large amount of intermediate data in memory. 

\begin{figure}[h]
\centering
\begin{subfigure}{0.48\linewidth}
  \centering
  \includegraphics[width=\textwidth]{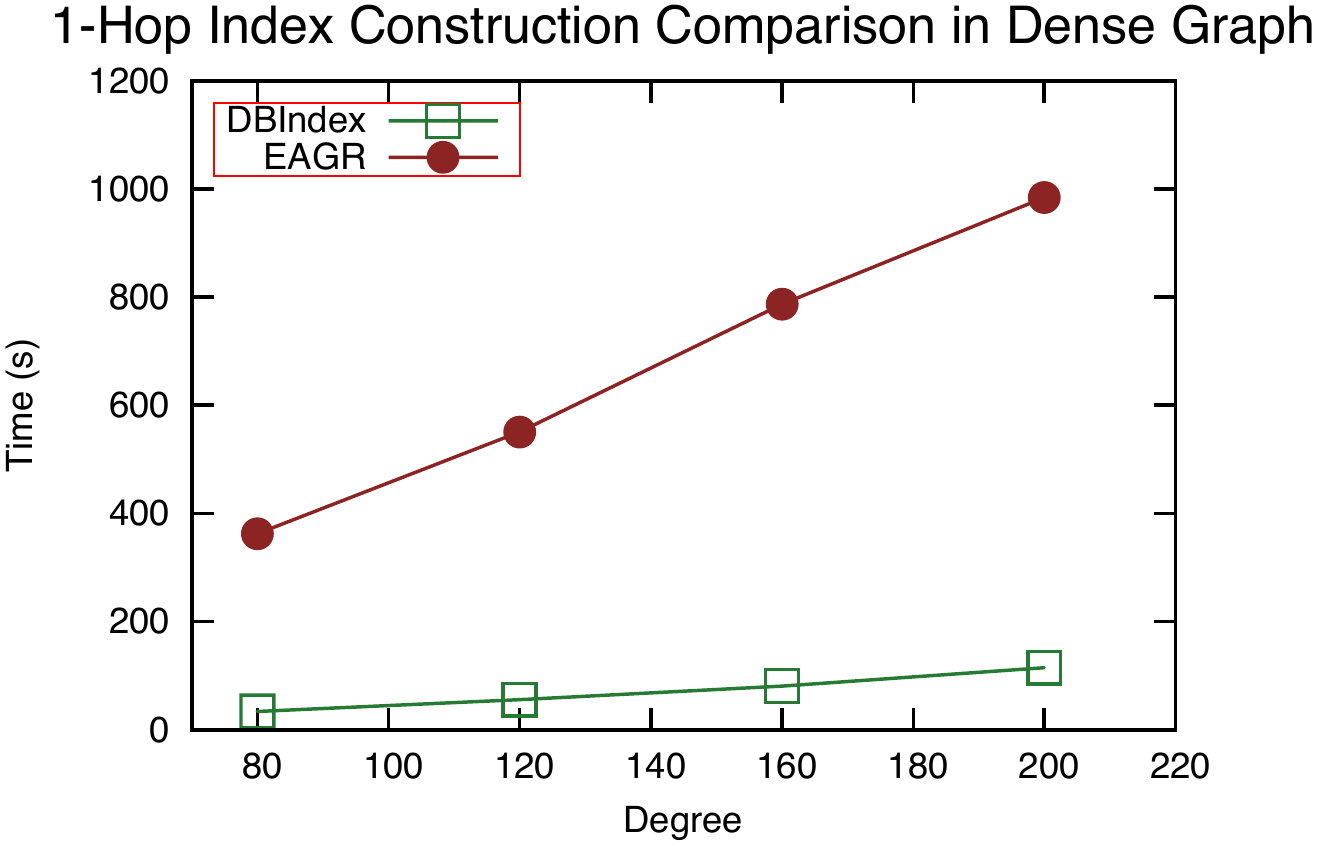}
  \caption{Index Construction}
\end{subfigure}
\begin{subfigure}{0.48\linewidth}
  \centering
  \includegraphics[width=\textwidth]{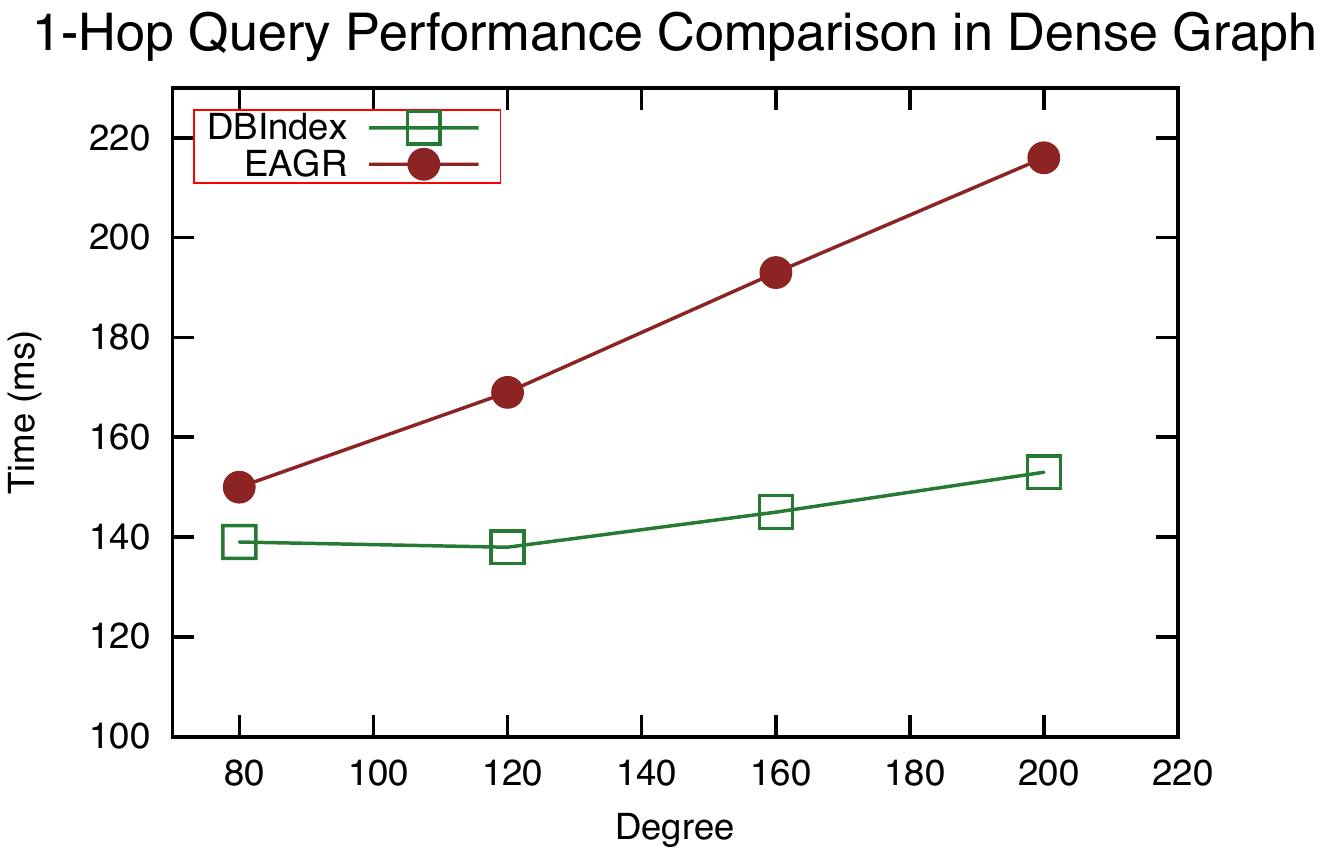}
  \caption{Query Performance}
\end{subfigure}
\begin{subfigure}{0.48\linewidth}
  \centering
  \includegraphics[width=\textwidth]{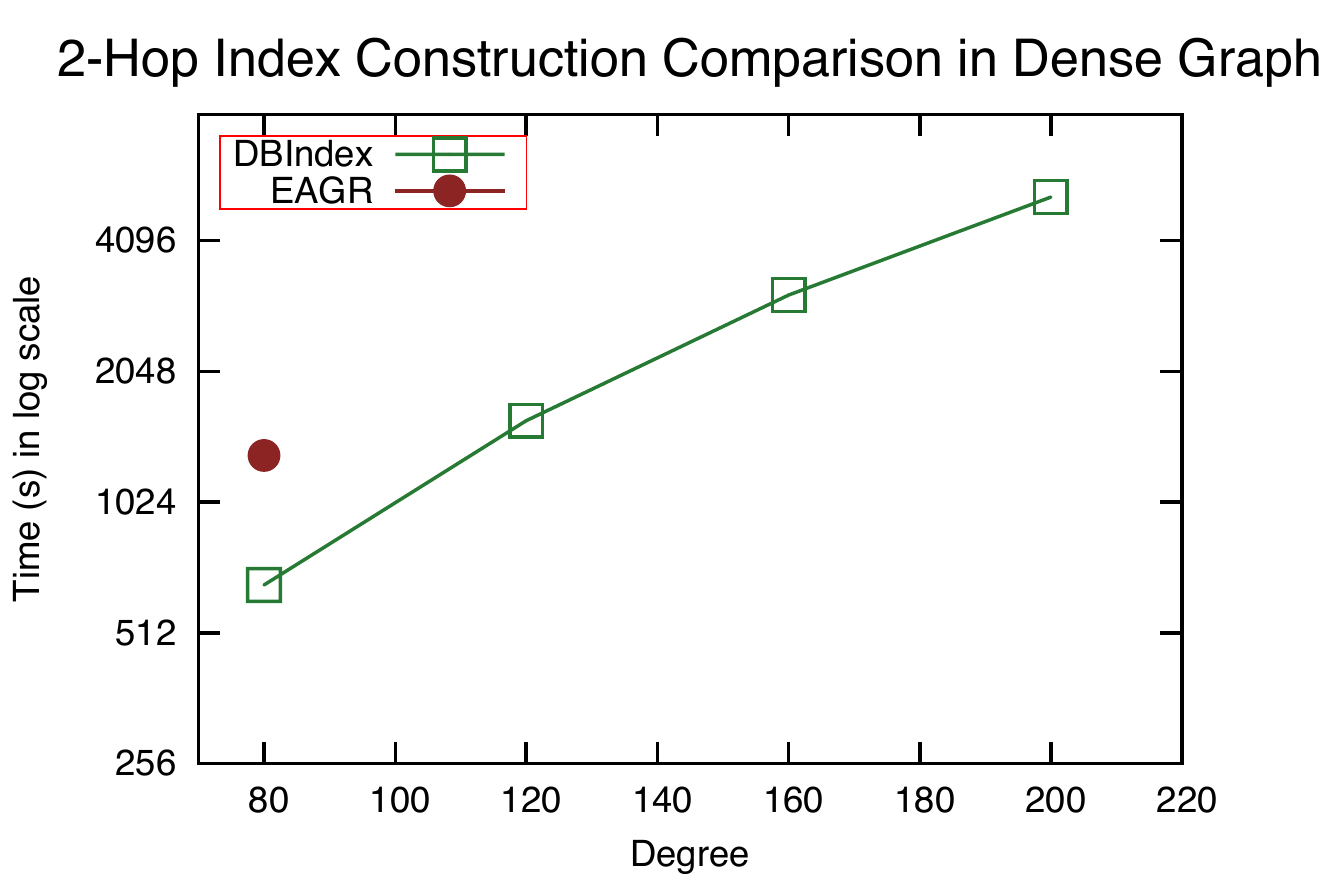}
  \caption{Index Construction}
\end{subfigure}
\begin{subfigure}{0.48\linewidth}
  \centering
  \includegraphics[width=\textwidth]{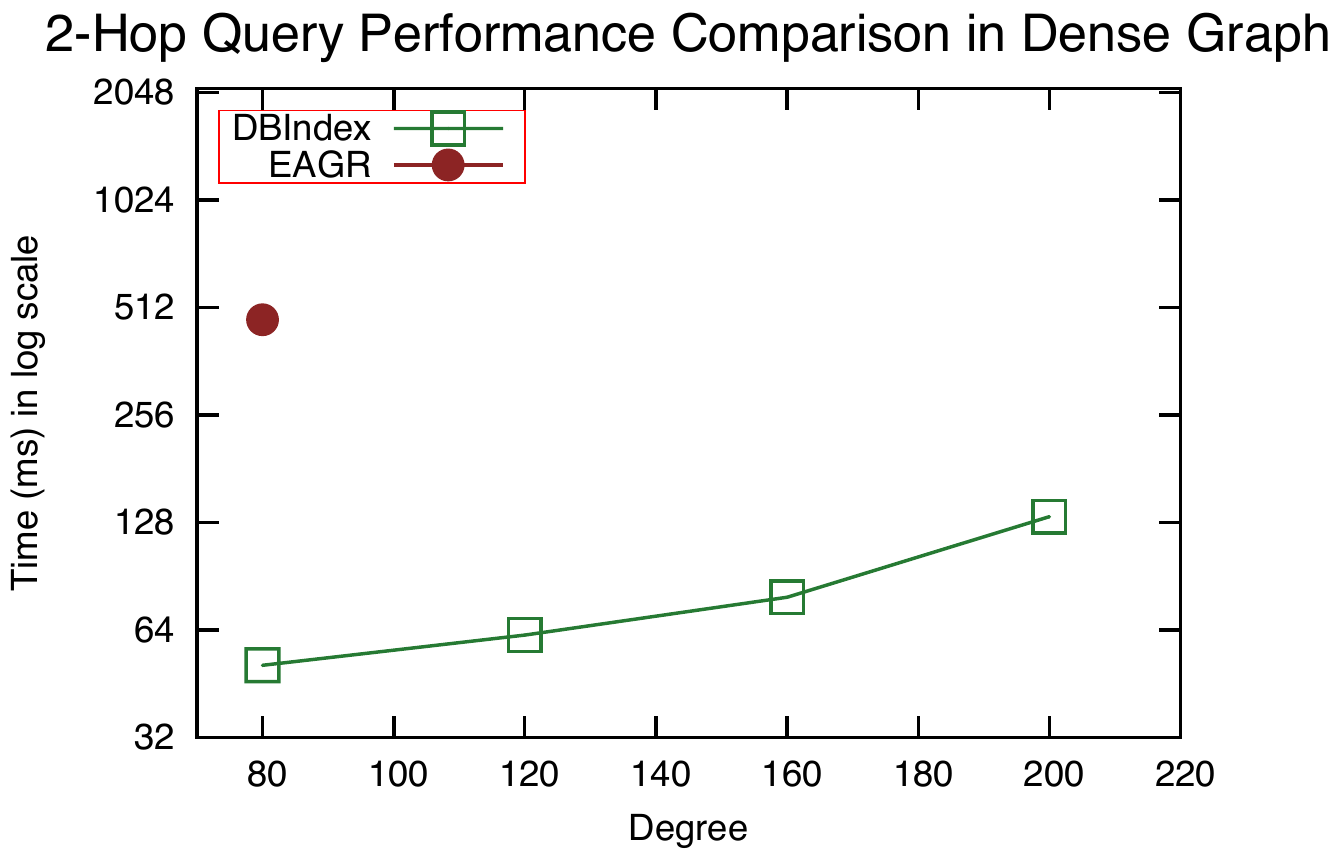}
  \caption{Query Performance }
\end{subfigure}
\caption{Impact of Degree over Dense Graphs over 200K Vertices. (a) and (b) are the results for 1-hop query; (c) and (d) are the results for 2-hop query. }
\label{fig:khop_v200k}
\end{figure}

\subsection{Evaluation of I-Index}
In this set of experiments, we evaluate I-Index. 
All the datasets are generated from the DAGGER generator.   

\textbf{Impact of Degree.} First, we evaluate the impact of degree 
changes when we fix the number of vertex as 30k and 60k.  
We compare \emph{DBIndex} with I-Index. 
In the query results, we also implement one non-index algorithm 
which dynamically calculates the window and then performs the aggregation. 
For indexing time, as shown in Figures~\ref{fig:pi_effect} (a) and (c), as the 
index size increases, both the indexing time of 
DBIndex and I-Index increase. However I-Index is more efficient than 
DBIndex, this is due to the special containment optimization used.  We observe that the 
index construction time is almost the same as the one time non-index query time. In other words, 
we can use one query time to create the index which is able to provide much faster query processing for subsequent queries. 
 In terms of query performance, shown in Figs.~\ref{fig:pi_effect} (b) and 
(d), the non-index approach is, on average, 20 times slower than 
the index-based schemes. I-Index outperforms DBIndex by 20\% to 30\%. 
The results clearly show that I-Index outperforms DBIndex for topological 
window in both index construction and query performance. Therefore, in the following experiments, we only present the results for I-Index. 
 
\begin{figure}
\centering
\begin{subfigure}{0.48\linewidth}
  \centering
  \includegraphics[width=\textwidth]{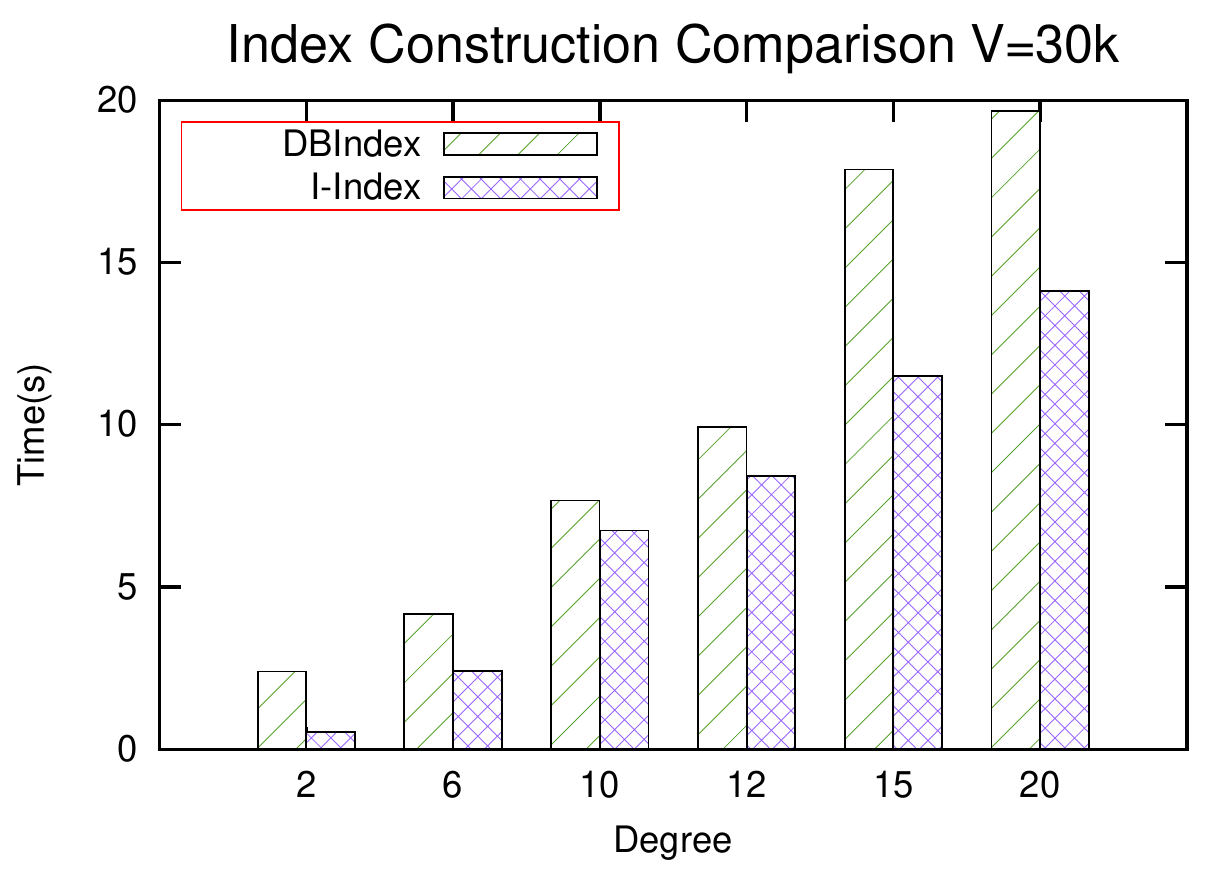}
  \caption{Index Construction}
\end{subfigure}
\begin{subfigure}{0.48\linewidth}
  \centering
  \includegraphics[width=\textwidth]{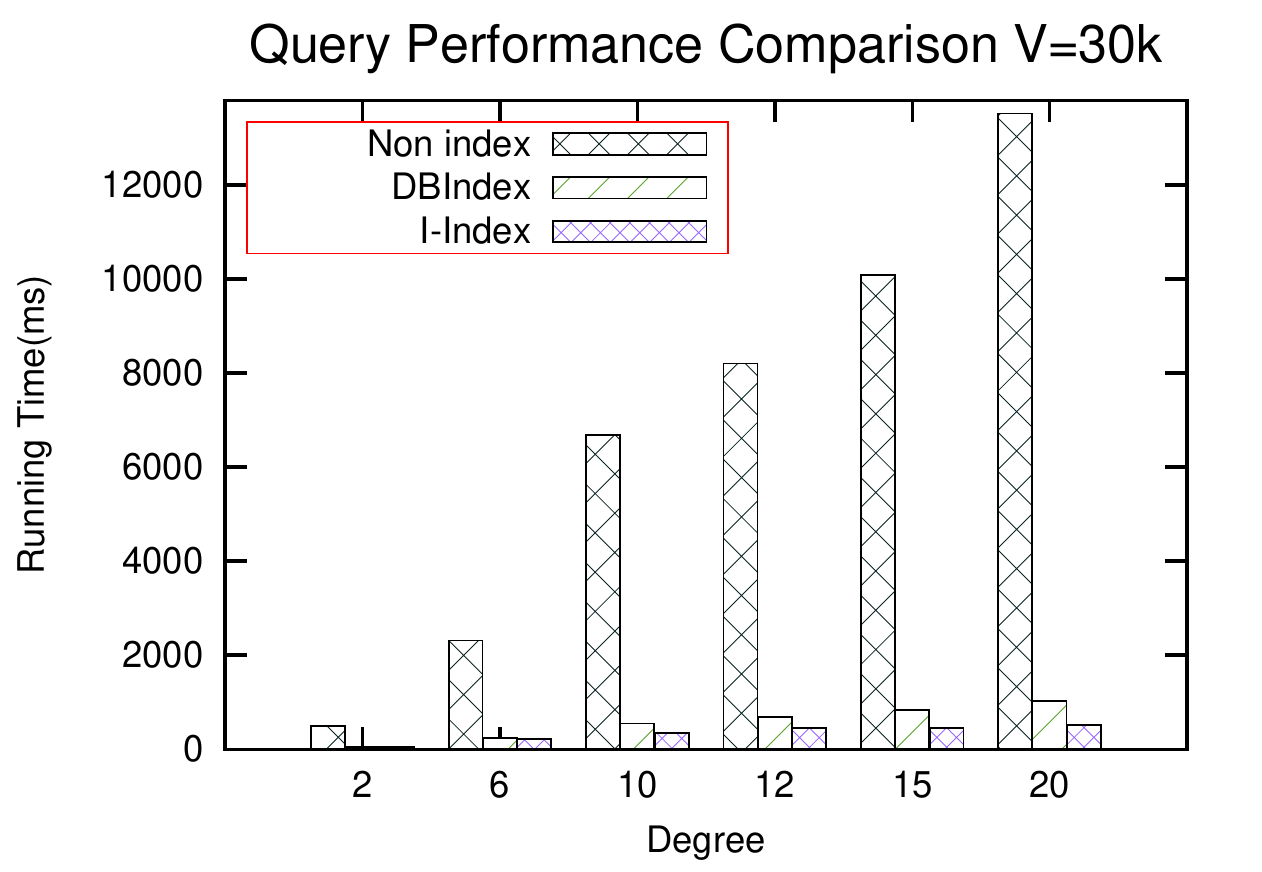}
  \caption{Query Performance}
\end{subfigure}
\begin{subfigure}{0.48\linewidth}
  \centering
  \includegraphics[width=\textwidth]{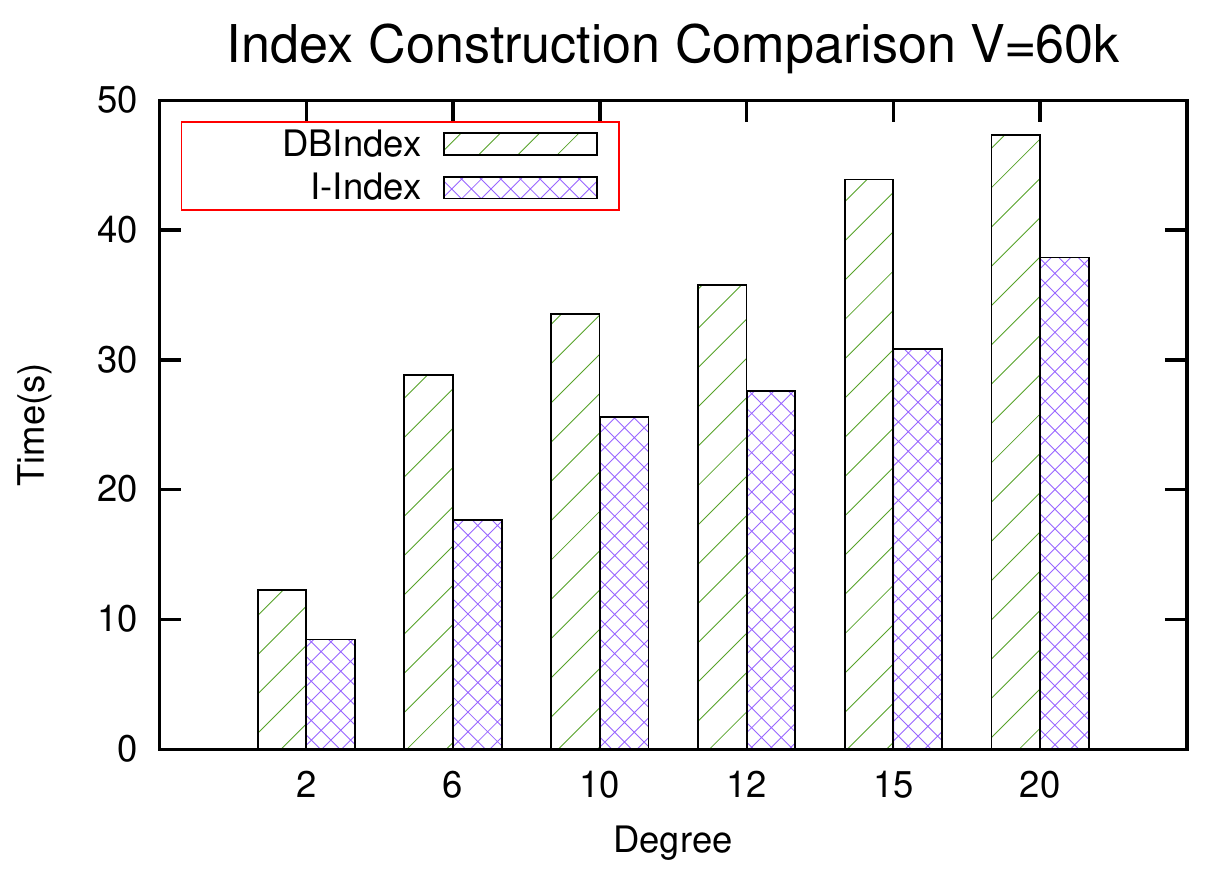}
  \caption{Index Construction}
\end{subfigure}
\begin{subfigure}{0.48\linewidth}
  \centering
  \includegraphics[width=\textwidth]{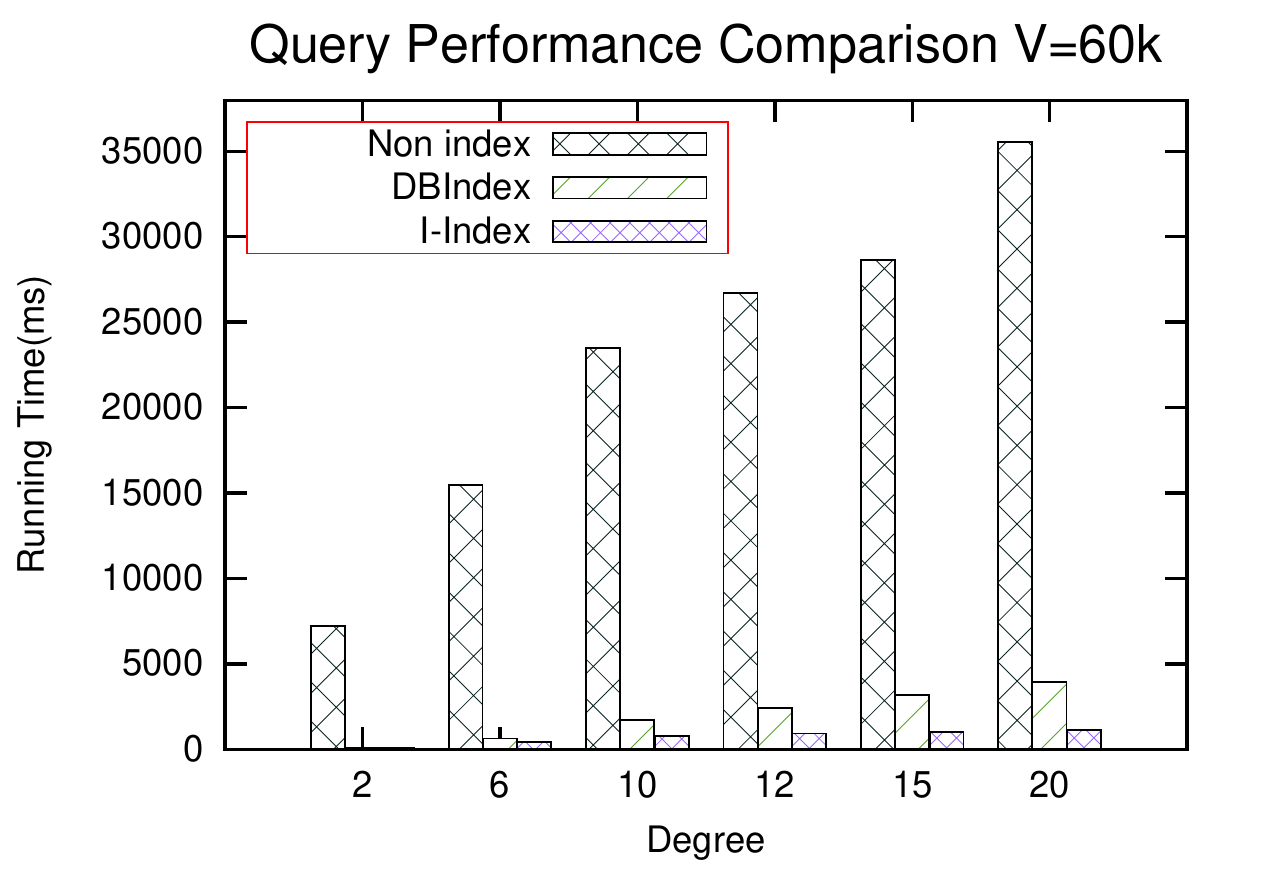}
  \caption{Query Performance}
\end{subfigure}
\caption{Impact of Degree. (a) and (b) are the results for 30K vertices; (c) and (d) are the results for 60K vertices.}
\label{fig:pi_effect}
\end{figure}

\textbf{Impact of Number of Vertices.} Next, we study how the performance of I-Index 
 is affected when we fix the degree and vary the number of vertices from 50k to 350K. 
Figs.~\ref{fig:pi_effect2} (a) and (c) show the index construction time when we fix the degree to 
10 and 20 respectively. From the results, we see that the construction time increases while the number of vertex increases and the construction time of a high degree graph is longer than that for low degree graphs. Figs.~\ref{fig:pi_effect2} (b) and (d) show the query time when we fix the degree to 10 and 20 respectively. 
As shown, the degree affects the query processing time - when the degree increases, the query time increases as well. We also observe that the query time is increasing linearly when the number of vertices increases. This shows the I-Index has good scalability.  

\textbf{Index Size.} Fig.~\ref{fig:top-index-size} presents
the index size ratio (i.e. size of index divided by the size of original graph) under different degrees from 3 to 20. 
There are four different sizes of data used with 100k, 150k, 200k and 300k vertices.  
For every vertex setting, the index size maintains the same trend in various degrees. The index size is linear to the input graph size. 
As a graph gets denser, the difference field of the I-Index
effectively shrinks. Thus, the index size in turn becomes smaller, which explains the bends in the figure.

\begin{figure}
\centering
\begin{subfigure}{0.48\linewidth}
  \centering
  \includegraphics[width=\textwidth]{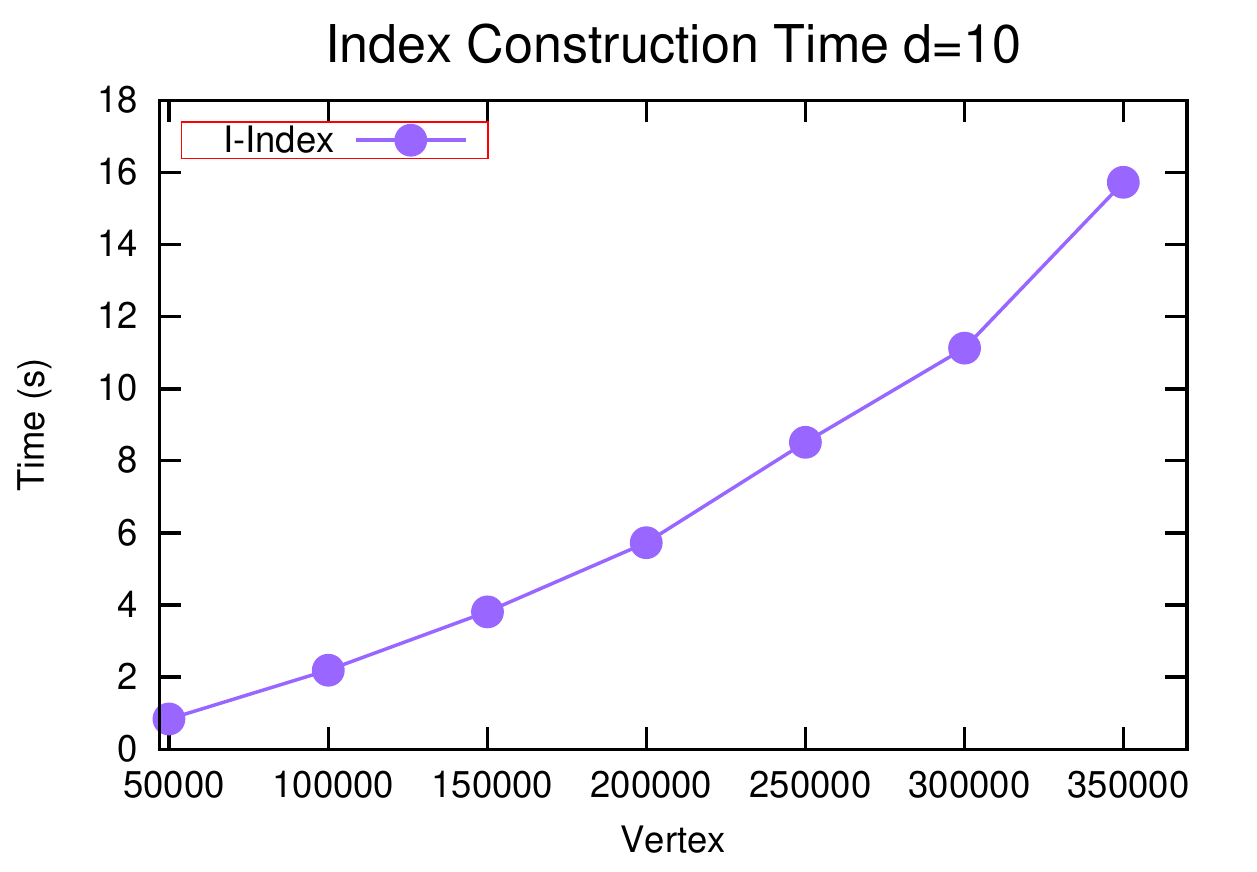}
  \caption{Index Construction }
\end{subfigure}
\begin{subfigure}{0.48\linewidth}
  \centering
  \includegraphics[width=\textwidth]{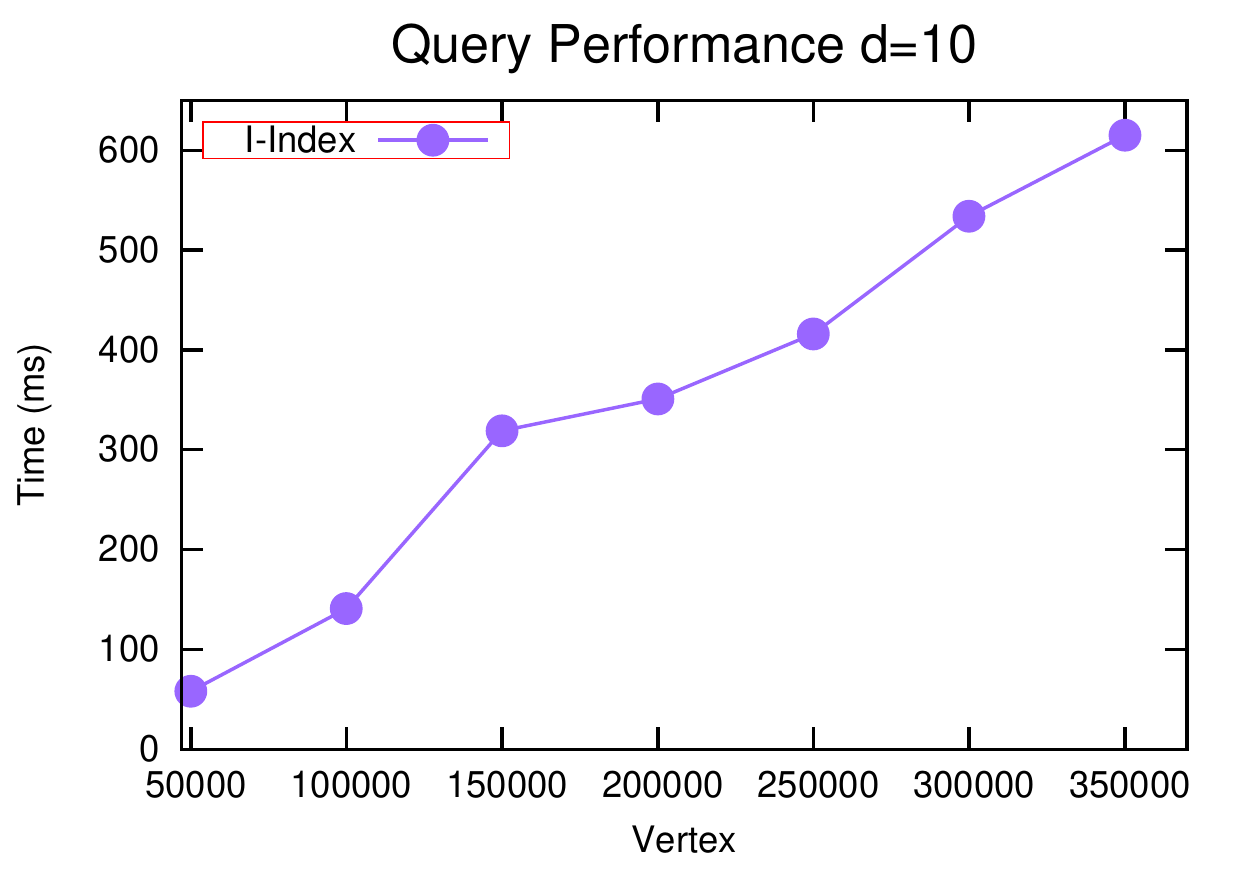}
  \caption{Query Performance}
\end{subfigure}
\begin{subfigure}{0.48\linewidth}
  \centering
  \includegraphics[width=\textwidth]{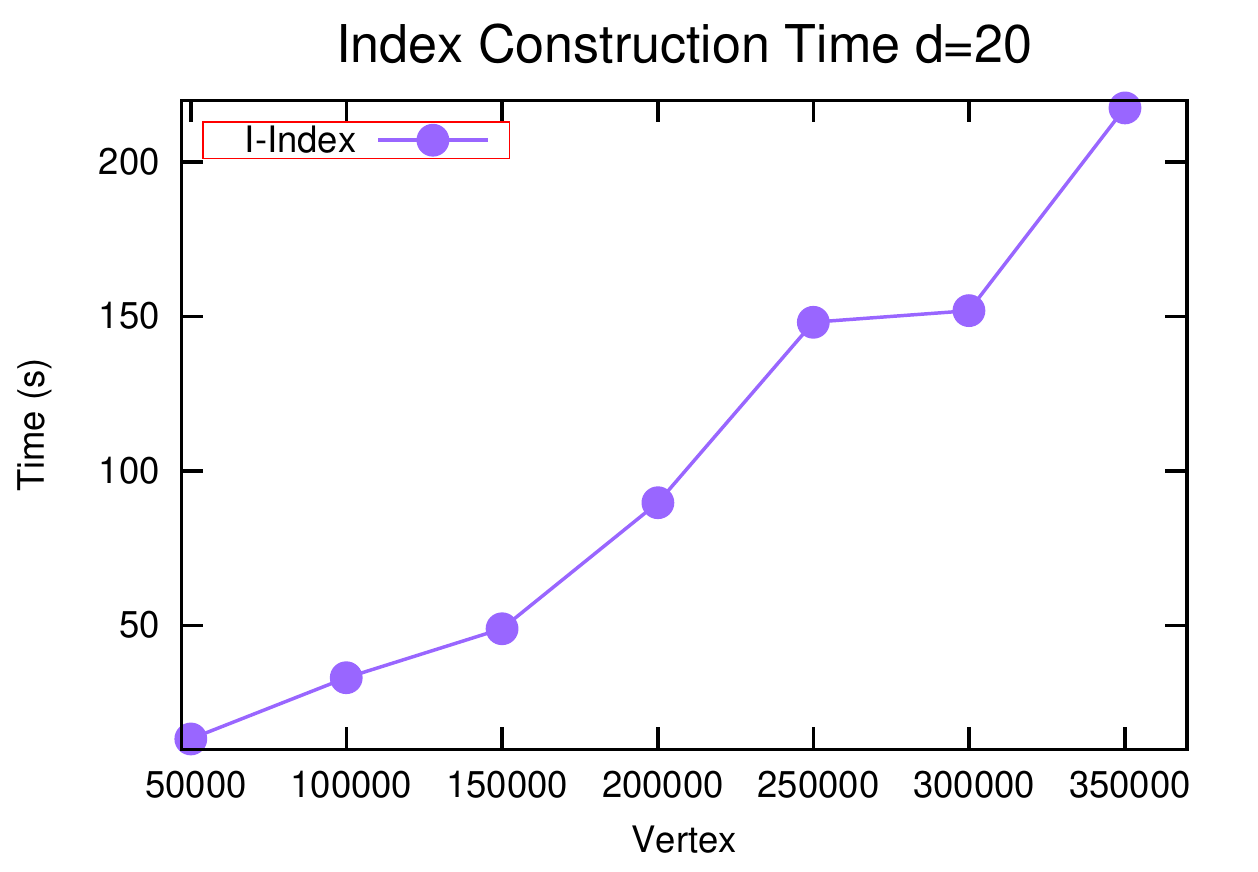}
  \caption{Index Construction}
\end{subfigure}
\begin{subfigure}{0.48\linewidth}
  \centering
  \includegraphics[width=\textwidth]{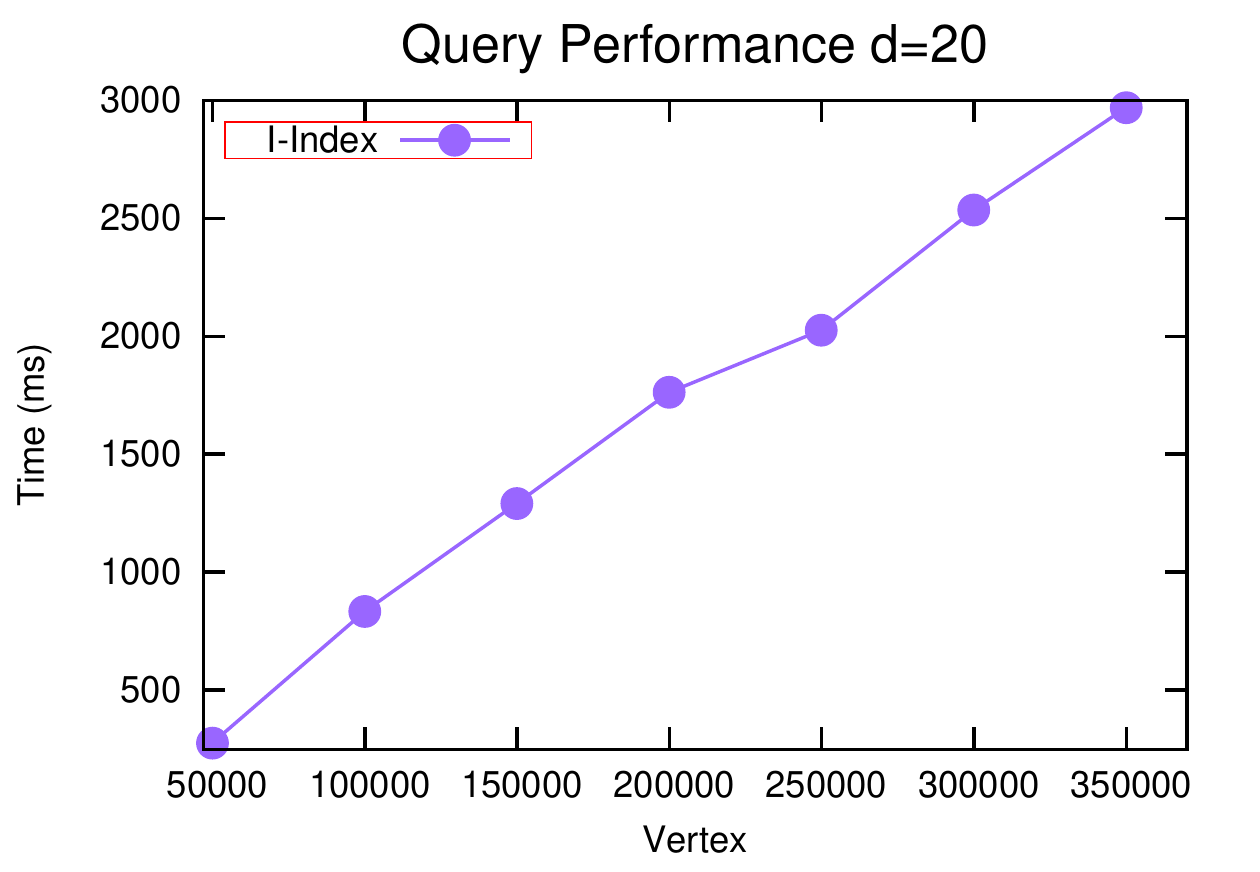}
  \caption{Query Performance}
\end{subfigure}
\caption{Impact of the number of vertices with a fixed degree. (a) and (b) 
are the results for the graphs with degree 10; (c) and (d) 
are the graphs with degree 20. }
\label{fig:pi_effect2}
\end{figure}

\begin{figure}[h]
\centering
\includegraphics[width=0.25\textwidth]{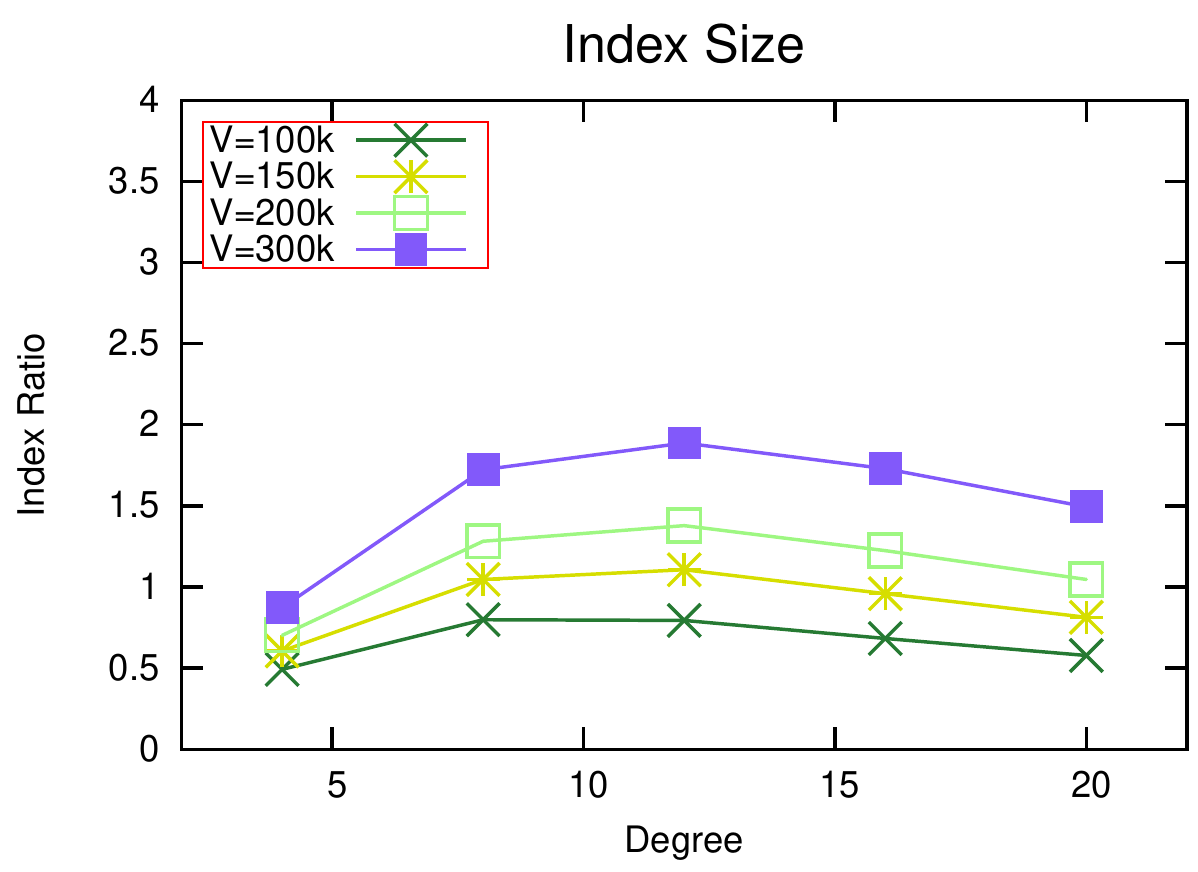}
\caption{Index Ratio of Inheritance-Index}
\label{fig:top-index-size}
\end{figure}

%% file: sec7_conclusion.tex
\section{Conclusion and Future Work}
In this paper, we have proposed a new type of graph analytic query,  \emph{Graph Window Query}. We formally defined two instantiations of graph windows: k-hop window and topological window.
We developed the Dense Block Index (DBIndex) to facilitate efficient processing of both types of graph windows. In addition, we also proposed the Inheritance Index (I-Index) that exploits a containment property of DAG to further improve the query performance of topological window queries. Both indices integrate window aggregation sharing techniques to salvage partial work done, which is both space and query efficient. We conducted extensive experimental evaluations over both large-scale real and synthetic datasets. The experimental results showed the efficiency and scalability of our proposed indices.

There remain many interesting research problems for graph window analytics. 
As part of our future work, we plan to explore structure-based window aggregations
which are complex than attri- bute-based window aggregations.
In structure-based aggregations, $W(v)$ refers to a subgraph of $G$ instead of a set of vertices,
and the aggregation function $\Sigma$ (e.g., centrality, PageRank, and {graph aggregation} \cite{wang2014pagrol,zhao2011graph}) operates on the structure of the subgraph $W(v)$. 